\documentclass[11pt,letterpaper]{article}
\usepackage[margin=1in]{geometry}

\usepackage{amsmath,amsthm}
\usepackage{amssymb,enumerate}
\usepackage{bm}
\usepackage[dvipsnames]{xcolor}
\usepackage{float}
\usepackage[T1]{fontenc}
\usepackage[pdftex]{graphicx}
\usepackage{comment}
\usepackage{multirow}
\usepackage{tabularx, booktabs}
\usepackage[normalem]{ulem}
\usepackage{tikz}
\usepackage{bm}
\usepackage[numbers, sort]{natbib}
\usepackage[export]{adjustbox}
\usepackage{pifont}
\usepackage{array}
\usepackage{caption}
\usepackage{subcaption}
\usepackage{crossreftools}

\usepackage{authblk}

\usepackage[pdftex,colorlinks=true,linkcolor=blue,citecolor=blue,urlcolor=blue]{hyperref}

\numberwithin{equation}{section}
\graphicspath{{./figures/}}

\newcommand{\bra}[1]{\langle #1 |}
\newcommand{\ket}[1]{| #1 \rangle}
\newcommand{\braket}[1]{\langle #1 \rangle}

\newcommand{\pder}[2][]{\frac{\partial #1}{\partial #2}}

\newcommand{\E}{\mathbb{E}}
\newcommand{\R}{\mathbb{R}}
\newcommand{\Z}{\mathbb{Z}}

\newcommand{\vect}[1]{\boldsymbol{#1}}
\newcommand{\av}{{\vect{a}}}
\newcommand{\bv}{{\vect{b}}}
\newcommand{\xv}{{\vect{x}}}
\newcommand{\yv}{{\vect{y}}}
\newcommand{\sv}{{\vect{s}}}

\newcommand{\xiv}{{\vect{\xi}}}
\newcommand{\paramv}{{\vect{\beta},\vect{\gamma}}}
\newcommand{\param}{{\beta,\gamma}}
\newcommand{\tN}{\textrm{N}}
\newcommand{\alg}[1]{{\texttt{#1}}}
\newcommand{\Id}{\vect{1}}

\newcommand{\G}[1]{{[#1]}}

\newcolumntype{Y}{>{\centering\arraybackslash}X}

\newtheorem{thm}{Theorem}
\newtheorem*{thm*}{Theorem}

\newtheorem{lem}[thm]{Lemma}
\newtheorem*{lem*}{Lemma}
\newtheorem{claim}[thm]{Claim}

\newtheorem{assump}[thm]{Assumption}
\newtheorem{defn}[thm]{Definition}

\theoremstyle{definition}
\newtheorem{rem}[thm]{Remark}

\numberwithin{thm}{section}

\DeclareMathOperator{\poly}{poly}

\begin{document}

\title{Quantum speedups in solving near-symmetric optimization problems by low-depth QAOA}

\author[1,2]{Ashley Montanaro\thanks{ashley@phasecraft.io}}
\affil{Phasecraft Ltd.}
\author[1,3]{Leo Zhou\thanks{leoxzhou@ucla.edu}}
\affil{University of Bristol}
\affil[3]{University of California, Los Angeles}

\date{February 24, 2025}
\maketitle

\begin{abstract}
We present new advances towards achieving exponential quantum speedups for solving optimization problems by low-depth quantum algorithms.
Specifically, we focus on families of combinatorial optimization problems that exhibit symmetry and contain planted solutions.
We rigorously prove that the 1-step Quantum Approximate Optimization Algorithm (QAOA) can achieve a success probability of $\Omega(1/\sqrt{n})$, and sometimes $\Omega(1)$, for finding the exact solution in many cases. 
This allows us to prove a separation of $O(1)$ quantum queries and $\Omega(n/\log n)$ classical queries required to find the planted solution in the latter setting.
Furthermore, we construct near-symmetric optimization problems by randomly sampling the individual clauses of symmetric problems, and prove that the QAOA maintains a strong success probability in this setting even when the symmetry is broken.
Finally, we construct various families of near-symmetric Max-SAT problems and benchmark state-of-the-art classical solvers, discovering instances where all known general-purpose classical algorithms require exponential time.
Therefore, our results indicate that low-depth QAOA may achieve an exponential quantum speedup for optimization problems.
\end{abstract}

\tableofcontents

\section{Introduction}
Optimization problems are critical to a wide range of real-world applications, and efficiently solving these problems is of great practical importance across many fields in science and industry. Quantum computers hold the promise of solving certain optimization problems faster than classical algorithms, offering potential breakthroughs in speed and efficiency. However, while there is hope for quantum speedups, we currently lack strong evidence that near-term quantum computers with limited circuit depth can achieve a substantial advantage over classical methods.

In this work, we provide new evidence that a low-depth quantum algorithm, the Quantum Approximate Optimization Algorithm (QAOA) \cite{farhi2014quantum}, can solve families of near-symmetric optimization problems exponentially faster than the best known classical algorithms. The QAOA has been proposed as a general-purpose quantum optimization algorithm that can be run on near-term quantum computers, and has seen implementation across a variety of experimental platforms~\cite{Pagano2020QAOA, harrigan2021quantum,Ebadi2022quantum, Shaydulin2024QAOA}.
In the low-depth regime, however, limitations of the QAOA has been proven for various problems~\cite{bravyi2020obstacles, farhi2020quantumWhole1, farhi2020quantumWhole2, chou2021limitations, basso2022performance, anshu2023concentration, chen2023local}.
While there has been some evidence that shows low-depth QAOA can still provide a quantum speedup for both approximate \cite{basso2021quantum} and exact optimization \cite{boulebnane2024SAT}, they appear to be only small polynomial speedups.
In contrast, our results demonstrate that for problems possessing some level of symmetry, an exponential speedup with the QAOA is possible even in the low-depth regime.

The symmetric optimization problems we consider are maximum constraint satisfaction problems (Max-CSPs) that are defined with a planted $n$-bit string and has a cost function that exhibit certain symmetry. For example, when the symmetry is the symmetric group $S_n$ permuting all $n$ bits, the cost function depends only on the Hamming distance to the planted bit string. Similar families of problems have been previously studied in the context of quantum annealing, which may take polynomial or exponential time to find the solution depending on the problem~\cite{vandam2001FOCS, farhi2002QAAvsSA, reichardt2004adiabatic, muthukrishnan2016}. Although a quantum speedup relative to certain general-purpose classical algorithms may be obtained in some cases, these symmetric problems are susceptible to attacks by tailored classical algorithms that know and take advantage of the symmetry. In this work, however, we take a step further to consider the situation where the symmetry is broken by randomly sampling the clauses in the cost function, and we call these problems ``near-symmetric.'' We also instantiate these problem explicitly in the form of Boolean satisfiability problems and benchmark them with state-of-the-art classical algorithms.

As our main result, we prove that the 1-step QAOA can solve many symmetric and near-symmetric problem in polynomial time.
Our findings are based on analytically deriving the success probability of 1-step QAOA in finding the planted solution using combinatorial calculations and a rigorous application of the saddle-point method.
For example, given any $S_n$-symmetric cost function that takes on $n+1$ distinct values for the different Hamming distances, we show the 1-step QAOA has $\Omega(1/\sqrt{n})$ success probability in finding the solution (Theorem~\ref{thm:distinct-values}).
We also consider various problems with either $S_n$ or $(S_{n/2})^2$ symmetry and multiple local minima that confuse classical algorithms, and prove that 1-step QAOA succeeds in finding the global minima with $\Omega(1)$ probability (Theorems~\ref{thm:unit-prob} and \ref{thm:prod-sym}).
Furthermore, leveraging the high success probability from a low-depth QAOA, we show that the QAOA also succeed with similar probability even when the symmetry is broken by random sparsification of the cost function (Theorem~\ref{thm:sparsified}).
Hence, using repetitions of the 1-step QAOA, one can find the solution in polynomial time with $O(\sqrt{n})$ or $O(1)$ queries to the cost function for these symmetric and near-symmetric problems.
On the other hand, we show that any classical algorithm requires at least $\Omega(n/\log n)$ queries to the cost function to solve any $S_n$-symmetric problem, even if the symmetry is known in advance (Claim~\ref{cl:alg-for-distinct}).

To explore quantum speedups, we explicitly construct instances of symmetric and near-symmetric Max-SAT problems, and study the performance of practical classical optimization algorithm such as simulated annealing and state-of-the-art classical SAT and Max-SAT solvers.
In particular, we consider algorithms that are front-runners from recent SAT and Max-SAT competitions, and perform numerical experiments up to hundreds of bits to extract their run time scaling.
For some instances, we show that all classical solvers known to us take exponential time to find the solution.
Since one can solve these problems in polynomial time using 1-step QAOA, we have what appears to be an exponential quantum speedup over general-purpose classical algorithms by low-depth QAOA.
Our results highlight the potential for significant quantum speedups even with constrained capabilities of near-term quantum devices, bringing us closer to realizing practical quantum advantages in optimization.

\section{Background}

\subsection{Quantum approximate optimization algorithm}
The QAOA is a quantum algorithm introduced by \cite{farhi2014quantum} for finding approximate solutions to combinatorial optimization problems.
The problem task is to minimize a cost function, which counts the number (or total weight) of clauses not satisfied by an input bit string.
Given a cost function $C(\xv)$ on bit strings $\xv \in \{0,1\}^n$,  we can define a corresponding quantum operator $C$, diagonal in the computational basis, as $C\ket{\xv} = C(\xv) \ket{\xv}$.
Moreover, we introduce the operator $B = \sum_{j=1}^n X_j$, where $X_j$ is the Pauli $X$ operator acting on qubit $j$.
Given a set of parameters $\vect\gamma = (\gamma_1,\gamma_2,\ldots,\gamma_p) \in \R^p$ and $\vect\beta = (\beta_1,\beta_2,\ldots,\beta_p) \in \R^p$, 
the $p$-step QAOA initializes the system of qubits in the uniform superposition of all bit strings, $\ket{s} = 2^{-n/2}\sum_{\xv}\ket{\xv}$, and applies $p$ alternating layers of unitary operations $e^{i\gamma_k C}$ and $e^{i\beta_k B}$ to prepare the state
\begin{align} \label{eq:wavefunction}
\ket{\psi_\paramv} = e^{i\beta_p B} e^{i\gamma_p C} \cdots e^{i\beta_1 B} e^{i\gamma_1 C} \ket{s}.
\end{align}
Finally, the state $\ket{\psi_\paramv}$ is measured in the computational basis to obtain a bit string $\xv$ with probability $|\braket{\xv|\psi_\paramv}|^2$.

Various strategies have been proposed to choose parameters $(\paramv)$ to minimize the expected value of the cost function $\braket{\psi_\paramv|C|\psi_\paramv}$ upon measurement for any given problem instance (see e.g., \cite{ZhouQAOA}). In this work, however, we will show how to choose the parameters so that we obtain  significant probability of finding an exact solution $\xv^*=\arg\min_\xv C(\xv)$, which may not always be the parameters that minimize the expected cost.

\subsection{Symmetric CSPs}
\label{sec:symmetric-CSP}
We aim to find exact solutions to combinatorial optimization problems in the form of maximum constraint satisfaction problems (Max-CSPs) with a planted solution.
These problems are described by a cost function $C:\{0,1\}^n \to \Z$  counting the number (or total weight) of violated constraints, such that $C(\xv)$ is minimized by the planted bit string $\sv \in \{0,1\}^n$, and $C(\xv) \ge C(\sv)$ for all $\xv \neq \sv$. 

Furthermore, we will consider symmetric optimization problems where the cost function is invariant under some symmetry group action on the input modulo the solution string. 
\begin{defn}[Symmetric CSPs] Given a permutation group  $G\subseteq S_n$ acting on $[n]$ and a bit string $\sv\in \{0,1\}^n$, we say that a Max-CSP with cost function $C: \{0,1\}^n\to \Z$ is \emph{$G$-symmetric relative to $\sv$} if $C(\pi(\xv\oplus\sv) \oplus \sv) = C(\xv)$ for all $\xv\in \{0,1\}^n$ and $\pi \in G$.
Here, $\oplus$ refers to bitwise modulo-2 addition, and $[\pi(\xv)]_i = x_{\pi(i)}$.
\end{defn}

Furthermore, we say that a given Max-CSP is a Max-$\ell$-CSP (or an $\ell$-CSP, for short) when the cost function is a sum of $\ell$-local constraint clauses, i.e., $C(\xv)=\sum_\alpha C_\alpha(\xv)$ where each $C_\alpha$ is a function that depends nontrivially on $\ell$ or fewer bits.
For concreteness, we will focus on the following two examples of symmetric CSPs and explain how to generate them by symmetrizing local constraints.

\paragraph{Example 1: $S_n$-symmetric CSPs.}
In this case, $C(\xv) = c(|\xv\oplus\sv|)$ for some function $c:\{0,\dots,n\} \to \Z$, where $|\xv\oplus\sv|$ is the Hamming distance between $\xv$ and $\sv$. In what follows, we will often refer to both $c$ and $C$ as the cost function interchangeably, and use the upper-case $C$ for bit string inputs and the lower-case $c$ for the Hamming distance inputs.

We can generate an $S_n$-symmetric CSP relative to $\sv$ from any local ``clause'' function $C_1:\{0,1\}^\ell \to \Z$ by applying $C_1$ to all $\ell$-subsets of the bits in $\xv\oplus\sv$.  If $C_1$ is not symmetric with respect to its input, we will consider all permutations of these subsets too, which effectively symmetrizes $C_1$.
For example, if $C_1 = x_1 \vee x_2$ is a boolean satisfiability clause, we would get an overall Max-SAT formula of the form $(y_1 \vee y_2) \wedge (y_1 \vee y_3) \wedge \dots \wedge (y_{n-1} \vee y_n)$ where $y_i=x_i\oplus s_i$, with a cost function given by the number of unsatisfied clauses. The final formula will have $\Theta(n^\ell)$ clauses.

As an example, let us consider the following SAT formula (with $\sv=0^n$ for simplicity):
\begin{equation} \label{eq:example-3SAT}
\phi = \phi_1 \wedge \phi_2, \qquad \text{where} \quad
\phi_1 = \lnot x_1 \wedge \lnot x_2 \wedge \dots \wedge \lnot x_n,
\quad
\phi_2 = \bigwedge_{i\neq j \neq k} (\lnot x_i \vee x_j \vee x_k).
\end{equation}
Here, $\phi_1$ ensures that the only satisfying assignment is $\sv=0^n$. $\phi_2$ ensures that most clauses are ones which encourage a local search algorithm to increase the Hamming weight of the current attempt. The cost function counting the number of unsatisfied clauses is $C(\xv)=\sum_i x_i + \sum_{i\neq j\neq k}x_i(1-x_j)(1-x_k)$. More simply, the cost of a bit string $\xv$ with Hamming weight $k$ is
\begin{equation} \label{eq:example-3SAT-cost}
c(k) = k + k (n-k)(n-k-1).
\end{equation}
Note that this cost function has two local minima: the global minimum at $k=0$ where $c(k)=0$, and a suboptimal local minimum at $k=n-1$ where $c(k)=n-1$. Moreover, the cost decreases in the range $k/n \ge 1/3$, which causes many classical algorithms to be trapped in the false minimum at $k=n-1$. For example, an algorithm based on hill climbing to decrease the cost, which starts with a uniformly random assignment, will go in the wrong direction. In addition, for most assignments, most of the violated clauses are ones in $\phi_2$; a local search algorithm like \alg{WalkSAT} that iteratively take a violated clause at random and flip one of the bits will more likely flip the last two bits to 1 instead of the first bit to 0, moving it farther from the solution.

\paragraph{Example 2: $S_{n_1}\times S_{n_2}$-symmetric CSPs with $n_1+n_2=n$.}
In this case, we partition any $n$-bit string into two substrings as $\xv=(\xv^\G1, \xv^\G2)$ where each $\xv^\G{i}$ is an $n_i$-bit string.
Then $C(\xv) = c( |\xv^\G1\oplus \sv^\G1|, |\xv^\G2 \oplus\sv^\G2| )$ for some function $c:[0, n_1]\times [0, n_2] \to \Z$.

Similar to the above example, we can generate an $S_{n_1}\times S_{n_2}$-symmetric CSP by taking any $(\ell_1 + \ell_2)$-local clause function $C_2:\{0,1\}^{\ell_1 + \ell_2} \to \Z$ and apply them to all possible unions of $\ell_1$-subsets from $n_1$ bits and and $\ell_2$-subset from $n_2$ bits. For example, if $(\ell_1,\ell_2)=(3,2)$ and $C_2 = x_1 (1- x_2)(1- x_3) x_4 (1-x_5)$, then the corresponding symmetrized CSP is
\begin{equation}
C(\xv) = \sum_{\substack{i,j,k\in [n_1]\\i\neq j\neq k}}~ \sum_{\substack{l,m\in [n_2] \\ l\neq m}} x^\G1_i (1-x^\G1_j)(1-x^\G1_k) x^\G2_l(1-x^\G2_m)
\end{equation}
Note upon an input with substring Hamming weights $k_1 = |\xv^\G1|$ and $k_2 = |\xv^\G2|$, the above cost function can be written more simply as
\begin{equation}
c(k_1, k_2) = k_1(n_1-k_1)(n_1-k_1-1) k_2(n_2-k_2-1).
\end{equation}

\subsection{Related work}

The performance of the QAOA has been studied for various optimization problems, including the Sherrington-Kirkpatrick model \cite{farhi2022quantum}, MaxCut \cite{basso2021quantum, boulebnane2021predicting}, the Max-$q$-XORSAT for regular hypergraphs \cite{basso2021quantum}, $q$-spin spin-glass models \cite{claes2021instance,basso2022performance}, and random Boolean satisfiability problems \cite{boulebnane2024SAT}.
Among these work, only \cite{boulebnane2024SAT} studied the QAOA in the context of exact optimization where the goal is to find the globally optimal solution. The results in \cite{boulebnane2024SAT} show that low-depth QAOA can find solutions to random $k$-SAT with exponentially small probability, although with a better exponent than the state-of-the-art classical SAT solvers that were benchmarked, yielding a polynomial quantum speedup.
Although \cite{basso2021quantum} showed that the QAOA provably achieves an approximate MaxCut value on random regular graphs better than any assumption-free polynomial-time classical algorithm, comparison against the conjectural state-of-the-art classical algorithm based on message passing~\cite{AMS2021MaxCut} suggests that a quadratic quantum speedup is more likely~\cite{TQC22talk}. Hence, the question of exponential quantum speedup in optimization via the QAOA is still open.

Another line of work has aimed to prove computational hardness results for the QAOA and related quantum algorithms. \cite{farhi2020quantumWhole1,farhi2020quantumWhole2, chou2021limitations, chen2023local} studied the limitation of local quantum algorithms like the QAOA for solving combinatorial optimization problems on sparse random graphs, using the bounded light-cone of the algorithms at sufficiently low depths. This limitation was later translated to the dense spin-glass models in \cite{basso2022performance}. In addition, \cite{bravyi2020obstacles, anshu2023concentration} proved hardness results for the QAOA by exploiting the bit-flip symmetry in certain families of problem.
However, all these previously known limitations of the QAOA have relied on concentration of the measured bit strings in the Hamming weight basis in some manner.
We remark that none of the existing hardness results apply to the settings we consider in the current paper, since the problems involve highly nonlocal geometry and the QAOA operates in a regime where there is no concentration of measurement.

There have also been multiple studies of quantum algorithms on symmetric optimization problem.
Early results in \cite{farhi2002QAAvsSA, vandam2001FOCS, reichardt2004adiabatic} examined adiabatic quantum annealing (i.e., the quantum adiabatic algorithm) on symmetric optimization problems where the cost function depends only on the Hamming weight of the input bit string.
In particular, \cite{farhi2002QAAvsSA} showed that for some of these problems with two local minima, classical simulated annealing takes exponential time to find the global minimum, whereas adiabatic quantum annealing may take polynomial or exponential time depending on the shape of the potential barrier between the local minima.
Later \cite{muthukrishnan2016} studied quantum annealing on more examples of such ``perturbed Hamming weight optimization'' problems, and found that short-time (non-adiabatic) quantum annealing can be much faster than adiabatic quantum annealing.
Nevertheless, many of these examples have cost functions that are highly nonlocal and cannot be instantiated as a Max-CSP with bounded locality. In contrast, all examples of symmetric CSPs we study in this work are made out of local clauses that are closer to practical optimization problems.

We remark that \cite{farhi2014quantum} has previously claimed but not proven the fact that the 1-step QAOA can solve a specific family of symmetric CSPs (which was first studied in \cite{farhi2002QAAvsSA}) that would take quantum annealing exponential time. Here, we develop techniques that enable a proof of the claim rigorously, as well as considering more general problems and situations when the symmetry is broken.

\section{QAOA can solve symmetric CSPs in polynomial time}
In this section, we show that the QAOA can solve symmetric CSPs in polynomial time for the two example symmetries $S_n$ and $S_{n_1}\times S_{n_2}$ that we introduced in Section~\ref{sec:symmetric-CSP}.

For $S_n$-symmetric CSPs, we first show in Theorem~\ref{thm:distinct-values} that 1-step QAOA can find the exact solution with probability $\Omega(1/\sqrt{n})$ under relatively general assumptions for the cost function.  Then we show in Theorem~\ref{thm:unit-prob} that the success probability of QAOA can be improved to $\Omega(1)$ when the cost function takes on a specific form. Although more restricted, the symmetric CSPs in the second case contain features that are hard for classical algorithms, which we explain in more detail in Section~\ref{sec:classical}.

Finally, we consider a generalization of our result to the product symmetry group $S_{n_1}\times S_{n_2}$ in Section~\ref{sec:prod-sym}. There, under a suitable constraint for the form of the cost function, we show in Theorem~\ref{thm:prod-sym} that 1-step QAOA can also achieve $\Omega(1)$ success probability in finding the exact solution for these more general problems.

\subsection{$S_n$-symmetric CSPs}
\label{sec:Sn-sym}

We begin by proposing a couple somewhat general assumptions for $S_n$-symmetric CSPs under which the QAOA succeeds to find the solution with probability $\Omega(1/\sqrt{n})$.
This is the assumption that the cost function takes on distinct values for different input Hamming distances.
More formally,
\begin{assump} \label{as:distinct-values}
$c(k)\neq c(l)$ for all $k\neq l$.
\end{assump}
\begin{assump} \label{as:approximate-distinct-values}
$\Pr_{\xv,\yv}\Big[c(|\xv|) = c(|\yv|) ~\text{and}~ |\xv|\neq |\yv|\Big]  = o(1/\sqrt{n})$ for bit strings $\xv,\yv \in \{0,1\}^n$ drawn uniformly at random.
\end{assump}
We believe that Assumption~\ref{as:distinct-values} can be easily checked to be satisfied for many $S_n$-symmetric CSPs constructed from symmetrizing local constraints.
Intuitively, an $S_n$-symmetric Max-$\ell$-CSP constructed from symmetrizing a $\ell$-local clause has a cost function  that takes on values in the range of $[0,\Theta(n^\ell)]$, and so it's usually unlikely that the $n+1$ values repeat.
More rigorously, we show for example later in Lemma~\ref{lem:distinct-values} that this assumption is satisfied for the symmetric SAT problem  given in Eq.~\eqref{eq:example-3SAT} with a simple number-theoretic proof.
Assumption~\ref{as:approximate-distinct-values} is weaker (and thus more general) since it is implied by Assumption~\ref{as:distinct-values}, and it allows $c(k)$ to take on repeated values for the relatively rarer Hamming distances.
Under either assumptions, we prove:
\begin{thm} \label{thm:distinct-values}
For any family of $S_n$-symmetric CSPs satisfying Assumption~\ref{as:distinct-values} or \ref{as:approximate-distinct-values}, the 1-step QAOA finds the solution with probability $\Omega(1/\sqrt{n})$.
\end{thm}

We remark that distinctness of the $c(k)$'s is required to some extent for the quantum algorithm to succeed in polynomial time. This is because in the extreme case when $c(0)=0$, $c(k) = 1$ for $k \neq 0$, we are in the situation of unstructured search, and the running time must be $\Omega(\sqrt{2^n})$ due to the lower bound by \cite{BBBV}.

Next, we consider the case of $S_n$-symmetric cost functions that take on a specific form and show that 1-step QAOA can succeed to find the solution with $\Omega(1)$ probability.
For any integers $\ell,a$ with $0\le a \le \ell$, we consider any cost function that looks like the following to leading order in $n$:
\begin{equation} \label{eq:symSATcost}
	c(k) = n^{\ell} \Big(\frac{k}{n}\Big)^a \Big(1-\frac{k}{n}\Big)^{\ell-a} + O(n^{\ell-1}).
\end{equation}
Note the example problem given earlier with cost function \eqref{eq:example-3SAT-cost} falls under this case. Furthermore, we expect such cost functions can be challenging for classical algorithms since they contain two distant local minima at $k=0$ and $k=n$.
For these problems, we show in the below theorem that there are parameters for 1-step QAOA to find the exact solution with $\Omega(1)$ probability.

\begin{thm}\label{thm:unit-prob}
For any $S_n$-symmetric problem with cost function of the form \eqref{eq:symSATcost} where $\ell\neq 2a$, the 1-step QAOA with $\gamma = \frac{2^{\ell-2}}{\ell-2a} \frac{\pi}{n^{\ell-1}}$ and $\beta=-\pi/4$ can achieve $\Omega(1)$ probability of finding the exact solution.
\end{thm}
In what follows, we prove Theorem~\ref{thm:distinct-values} in Section~\ref{sec:proof-distinct-values}, and then Theorem~\ref{thm:unit-prob} in Section~\ref{sec:proof-unit-prob}.

\subsubsection{Success when cost function takes on (mostly) distinct values}
\label{sec:proof-distinct-values}

To prove Theorem \ref{thm:distinct-values} under either Assumption~\ref{as:distinct-values} or \ref{as:approximate-distinct-values}, we begin by noting that the $1$-step QAOA produces the following state for parameters $\beta,\gamma \in \R$:
\[ \ket{\psi_\param} = e^{i\beta \sum_j X_j} e^{i\gamma C} \ket{+}^{\otimes n}.
\]
Observe that
\begin{eqnarray*}
	\braket{\sv|\psi_\param} &=& \bra{0^n} \prod_{j,s_j=1} X_j e^{i\beta \sum_j X_j} e^{i\gamma C} \ket{+}^{\otimes n}= \frac{1}{\sqrt{2^n}}\bra{0^n} e^{i\beta \sum_j X_j} \sum_{\xv \in \{0,1\}^n} e^{i \gamma c(|\xv\oplus\sv|)} \ket{\xv\oplus \sv}\\
	&=& \frac{1}{\sqrt{2^n}}\bra{0^n} e^{i\beta \sum_j X_j} \sum_{\xv \in \{0,1\}^n} e^{i \gamma c(|\xv|)} \ket{\xv}
\end{eqnarray*}
so without loss of generality we can assume that $\sv=0^n$. Then let us fix $\beta=-\pi/4$ and compute
\begin{eqnarray}
\braket{\sv|\psi_\param} &=& \frac{1}{2^n} (\bra{0}+i\bra{1})^{\otimes n} \sum_{\xv \in \{0,1\}^n} e^{i \gamma c(|\xv|)} \ket{\xv} = \frac{1}{2^n} \sum_{\xv \in \{0,1\}^n} i^{|\xv|} e^{i \gamma c(|\xv|)} \nonumber\\
	&=& \frac{1}{2^n} \sum_{k=0}^n i^{k} e^{i \gamma c(k)} \binom{n}{k}.
 \label{eq:overlap-general}
\end{eqnarray}

\begin{proof}[Proof of Theorem~\ref{thm:distinct-values}]
Let us choose $\gamma$ uniformly at random from $[0,2\pi)$ and compute the expected success probability. Starting from Eq.~\eqref{eq:overlap-general}, we get
\begin{eqnarray}
\E_\gamma[|\braket{\sv|\psi_\param}|^2] &=& \frac{1}{2^{2n}} \sum_{k,l=0}^n i^{k-l} \E_\gamma[e^{i \gamma (c(k)-c(l))}] \binom{n}{k} \binom{n}{l} = \frac{1}{2^{2n}} \sum_{k=0}^n \binom{n}{k}^2  + \epsilon_n =  \frac{1}{2^{2n}} \binom{2n}{n} + \epsilon_n \nonumber\\
	&= & \Theta(1/\sqrt{n}) + \epsilon_n
\end{eqnarray}
where we used Stirling's approximation for the first term. The error term collects the remaining part of the initial double sum:
\begin{align*}
\epsilon_n := \frac{1}{2^{2n}} \sum_{k\neq l} \binom{n}{k} \binom{n}{l}  i^{k-l}\E_\gamma[e^{i \gamma (c(k)-c(l))}] = \frac{1}{2^{2n}} \sum_{k\neq l} \binom{n}{k} \binom{n}{l}  i^{k-l}\Id_{c(k)=c(l)},
\end{align*}
where we used the fact that $\E_\gamma[e^{i \gamma z}] = 0$ for any integer $z \neq 0$, and denoted $\Id_A$ as the indicator function that returns 1 if statement $A$ is true and 0 otherwise.
For cost functions satisfying Assumption~\ref{as:distinct-values}, we have $\epsilon_n=0$, which immediately proves the theorem for this case.
For the more general case under Assumption~\ref{as:approximate-distinct-values}, let us rewrite the error term as
\begin{align}
\epsilon_n &= \frac{1}{2^{2n}} \sum_{\xv,\yv\in\{0,1\}^n}   i^{|\xv|-|\yv|} \Id_{c(|\xv|)=c(|\yv|) \,\wedge\, |\xv|\neq |\yv|} \nonumber\\
|\epsilon_n| &\le \frac{1}{2^{2n}} \sum_{\xv,\yv\in\{0,1\}^n}  \Id_{c(|\xv|)=c(|\yv|) \,\wedge\, |\xv|\neq |\yv|} = \Pr_{\xv,\yv}\Big[c(|\xv|)=c(|\yv|) \,\wedge\, |\xv|\neq |\yv|\Big].
\end{align}
Assumption~\ref{as:approximate-distinct-values} implies $|\epsilon_n|=o(1/\sqrt{n})$, which implies $\E_\gamma[|\braket{\sv|\psi_\param}|^2]=\Theta(1/\sqrt{n})$.
Therefore, under either assumptions, QAOA outputs the correct answer $\sv$ with probability $\Theta(1/\sqrt{n})$, and $O(\sqrt{n})$ repetitions are sufficient to identify $\sv$.
\end{proof}

To illustrate the applicability of Theorem~\ref{thm:distinct-values}, we show in the below lemma that the example $S_n$-symmetric SAT problem given in Section~\ref{sec:symmetric-CSP} indeed satisfies the distinct-value assumption, and hence the theorem can be applied to rigorously prove the success of 1-step QAOA.

\begin{lem} \label{lem:distinct-values}
Consider the SAT problem defined in Eq.~\eqref{eq:example-3SAT} with cost function $c(k)$ in Eq.~\eqref{eq:example-3SAT-cost}.
Then for all $n$ satisfying $n\equiv 0$ or $1\mod 4$, Assumption~\ref{as:distinct-values} is satisfied.
\end{lem}
\begin{proof}
	Suppose for the sake of contradiction that $c(k)=c(m)$ for some $k\neq m$. Then
	\begin{equation}
		0 = \frac{c(m)-c(k)}{m-k} = k^2 + k(m+1-2n) + (m-n)^2 + m-n + 1.
	\end{equation}
	To simplify this expression, let $\delta = m-n$, which allows us to write
	\begin{equation} \label{eq:diophantine-simple}
		0 =	k^2 + k(\delta+1-n) + \delta^2+\delta+1.
	\end{equation}
	If this equation has an integer solution, that it must have a solution in mod 4 as well, yielding
	\begin{equation*}
		0 =	k^2 + k(\delta+1-n) + \delta^2+\delta+1 \mod 4.
	\end{equation*}
	It is straightforward to verify that when $n=0$ or $1\mod4$, the above equality is not satisfiable for all $4^2$ possible of choices of $(k,\delta) \mod 4$. Thus, no integer solution to Eq.~\eqref{eq:diophantine-simple} exists, implying there is no integer solution to $c(m)=c(k)$ where $k\neq m$.
\end{proof}

\subsubsection{Success in some cases with optimized $\gamma$}
\label{sec:proof-unit-prob}

We now prove Theorem~\ref{thm:unit-prob} to show that 1-step QAOA can succeed with $\Omega(1)$ probability on some families of $S_n$-symmetric CSPs.

We will again set $\beta=-\pi/4$ and consider the parameter regime of $\gamma$ where 
\begin{align*}
	\gamma = \frac{\Gamma}{n^{\ell-1}} \pi.
\end{align*}
Following Eq.~\eqref{eq:overlap-general} and plugging in the form of $c(k)$ assumed in Eq.~\eqref{eq:symSATcost}, we have
\begin{equation}
\braket{\sv|\psi_\param} = \frac{1}{2^n} \sum_{k=0}^n \binom{n}{k} \exp\Big[{in \frac{\pi}{2}  P_{n}(k/n)}\Big] \label{eq:amp-sum}
\end{equation}
where
\[
P_n(\xi) = \xi + 2\Gamma \xi^a  (1-\xi)^{\ell-a} + O\Big(\frac{1}{n}\Big).
\]
Let us now rewrite the sum in \eqref{eq:amp-sum} as an integral. By Stirling's approximation, we have
\[
\binom{n}{k} = \sqrt{\frac{n}{2\pi k(n-k)}} \frac{n^n}{k^k(n-k)^{n-k}}  \times (1+o(1)),
\]
where  $o(1)$ is an error that vanishes with $n\to\infty$.
Let
\begin{equation}
S(\xi)	= -\xi\log \xi - (1-\xi)\log(1-\xi)  - \log 2 +  i\frac{\pi}{2} \Big[\xi + 2\Gamma \xi^a  (1-\xi)^{\ell-a}\Big],
\end{equation}
then
\begin{equation}
    \braket{\sv|\psi_\param} = \sum_{k=0}^n  \frac{1}{\sqrt{2\pi k(1-k/n)}} e^{n S(k/n)+ i\epsilon_n(k/n)} \times (1+o(1)) \label{eq:Riemann-sum}
\end{equation}
where $\epsilon_n(\xi)$ contains the left-over terms in $P_n(\xi)$ that remains uniformly bounded as $n\to\infty$.
We next apply the Euler-Maclaurin formula to approximate Eq.~\eqref{eq:Riemann-sum} with an integral while neglecting the exponentially small boundary terms as additional errors (which will be justified when we soon show that the final result is order unity), yielding
\begin{align}
\braket{\sv|\psi_\param} &= \int_0^n dk \frac{1}{\sqrt{2\pi k(1-k/n)}} e^{n S(k/n)+ i\epsilon_n(k/n)} \times (1+o(1))  \nonumber \\
&=
\int_0^1 d\xi \frac{\sqrt{n}}{\sqrt{2\pi \xi(1-\xi)}} e^{n S(\xi)+ i\epsilon_n(\xi)} \times (1+o(1)),
\label{eq:integral}
\end{align}
where we changed integration variable to $\xi=k/n$.

We note that  $\Re[S(\xi)]$ has a single maximum over $\xi\in(0,1)$ at $\xi_0=1/2$. We want to apply the saddle-point method to evaluate Eq.~\eqref{eq:integral} which requires $\xi_0$ to be a non-degenerate saddle point of $S(\xi)$. This requirement translates to the conditions that
\begin{equation}
0 = S'(\xi_0)= i\frac{\pi}{2} \Big(1 + \frac{2a-\ell}{2^{\ell-2}}\Gamma\Big),
\quad \text{and}\quad
0 \neq S''(\xi_0) = -4 + i\pi \Gamma \frac{(\ell-2a)^2 - \ell}{2^{\ell-2}}.
\end{equation}
These conditions can be satisfied when $\ell\neq 2a$ by choosing
\begin{equation*}
	\Gamma = \frac{2^{\ell-2}}{\ell-2a},
\end{equation*}
which ensures that $\xi_0$ is a unique non-degenerate saddle point of $S(\xi)$. Then the integral in Eq.~\eqref{eq:integral} satisfies the prerequisites for applying the saddle-point method (see e.g., \cite{PinnaViola2019}), yielding
\begin{align}
\braket{\sv|\psi_\param} &=  \frac{\sqrt{n}}{\sqrt{2\pi /4}} \sqrt{-\frac{2\pi}{n S''(\xi_0)}} e^{n S(\xi_0) + i\epsilon_n(\xi_0)} \times (1+o(1)), \nonumber \\
\lim_{n\to\infty} \left|\braket{\sv|\psi_\param}\right|^2 &=  \left| \frac{4}{4-i\pi \frac{(\ell-2a)^2-\ell}{\ell-2}}  \right| = \bigg(1+\pi^2\frac{[(\ell-2a)^2-\ell]^2}{16(\ell-2)^2} \bigg)^{-1/2}.
\end{align}
Note we used the fact that $\Re[S(\xi_0)]=0$ and $\epsilon_n(\xi_0) \in \R$.

\subsection{$S_{n_1}\times S_{n_2}$-symmetric CSPs}
\label{sec:prod-sym}

In this section, we study families of $S_{n_1}\times S_{n_2}$-symmetric CSPs in the large problem size limit. Specifically, we will let $\alpha_i = n_i/n$ and consider the limit where $n\to\infty$ with $\alpha_i$ as fixed constants.

For any $S_{n_1}\times S_{n_2}$-symmetric CSPs, the cost function can be written as 
\begin{align}
C(\xv) = c(d_1(\xv,\sv), d_2(\xv,\sv))
\end{align}
where $d_j(\xv, \sv)= |\xv^\G{j} \oplus \sv^\G{j} |$ is the Hamming distance between $\xv$ and $\sv$ among the subset of $n_j$ bits.

Similar to the setting of Theorem~\ref{thm:unit-prob}, we will consider situations where the cost function take on a restricted form where more explicit calculations can be done.
\begin{assump} \label{as:prod-sym-form}
Consider a family of $S_{n_1}\times S_{n_2}$-symmetric CSPs in the limit of $n\to\infty$ with $\alpha_i = n_i/n >0$ fixed.
Assume there is a positive integer $\ell$, a finite sequence of non-negative integer 4-tuples $(a_\mu,b_\mu,c_\mu, d_\mu)$ satisfying $a_\mu+b_\mu+c_\mu+d_\mu = \ell$ for all $\mu$, and a sequence of $\kappa_\mu\in \R$, such that the cost function of the CSP can be written to leading order in $n$ as
\begin{gather}
c(k_1, k_2) = n^\ell  f(k_1/n_1, k_2/n_2) + O(n^{\ell-1}),
\nonumber \\
\text{where} \quad
f(\xi_1, \xi_2) = \sum_{\mu} \kappa_\mu \xi_1^{a_\mu} (1-\xi_1)^{b_\mu} \xi_2^{c_\mu} (1-\xi_2)^{d_\mu}.
\label{eq:prod-sym-cost}
\end{gather}
Furthermore, we assume that
\begin{gather}
\frac{1}{\alpha_1} \sum_\mu \kappa_\mu (a_\mu - b_\mu) = \frac{1}{\alpha_2} \sum_\mu \kappa_\mu (c_\mu - d_\mu) \neq 0,
\label{eq:grad-cond}
\\
\text{and} \quad 
\det
\begin{pmatrix}
4\alpha_1 + i\pi\alpha_1 \frac{\sum_\mu\kappa_\mu [(a_\mu-b_\mu)^2 - (a_\mu+b_\mu)]}{\sum_\mu \kappa_\mu(a_\mu - b_\mu)}  & i\pi\alpha_1 \frac{\sum_\mu \kappa_\mu (a_\mu - b_\mu)(c_\mu - d_\mu)}{\sum_\mu \kappa_\mu(a_\mu-b_\mu)}  \\
i\pi\alpha_1 \frac{\sum_\mu \kappa_\mu (a_\mu - b_\mu)(c_\mu - d_\mu)}{\sum_\mu \kappa_\mu(a_\mu-b_\mu)} & 4\alpha_2 + i\pi\alpha_2 \frac{\sum_\mu \kappa_\mu [(c_\mu-d_\mu)^2 - (c_\mu+d_\mu)]}{\sum_\mu \kappa_\mu(c_\mu - d_\mu)}
\end{pmatrix} \neq 0.
\label{eq:hessian-cond}
\end{gather}
\end{assump}
We remark that one method to ensure that conditions \eqref{eq:grad-cond} and \eqref{eq:hessian-cond} are satisfied, for example, is by choosing $\alpha_1=\alpha_2=1/2$, having at least one term where $a_\mu\neq b_\mu$, and then symmetrizing the cost function over the two subsets to get $c'(\xi_1,\xi_2)= c(\xi_1,\xi_2) + c(\xi_2,\xi_1)$.

\begin{thm}\label{thm:prod-sym}
Consider any $S_{n_1}\times S_{n_2}$-symmetric CSP satisfying Assumption~\ref{as:prod-sym-form}.
In the limit of $n\to\infty$ with $n_i/n=\alpha_i$ held fixed,  the 1-step QAOA can achieve $\Omega(1)$ probability of finding the exact solution with $\gamma=\Theta(1/n^{\ell-1})$.
\end{thm}

\begin{proof}
We again set $\beta=-\pi/4$ and consider the parameter regime of $\gamma$ where
\begin{equation}
\gamma = \frac{\Gamma}{n^{\ell-1}} \pi.
\end{equation} Similar to the $S_n$-symmetric case, we can write the overlap of the 1-step QAOA state with the solution bit string as
\begin{align}
\braket{\sv|\psi_\param} &= \frac{1}{2^n} (\bra{0}+i\bra{1})^{\otimes n} \sum_{\xv \in \{0,1\}^n} e^{i \gamma c( d_1(\xv,\sv), d_2(\xv,\sv))} \ket{\xv\oplus\sv} \nonumber= \frac{1}{2^n} \sum_{k_1=0}^{n_1}\sum_{k_2=0}^{n_2} i^{k_1+k_2} e^{i \gamma c(k_1, k_2)} \binom{n_1}{k_1} \binom{n_2}{k_2}  \\
&= \frac{1}{2^n} \sum_{k_1=0}^{n_1}\sum_{k_2=0}^{n_2}  \binom{n_1}{k_1} \binom{n_2}{k_2} \exp\Big[i n\frac{\pi}{2} P_n\Big(\frac{k_1}{n_1}, \frac{k_2}{n_2}\Big) \Big] ,
\label{eq:prod-sym-sum}
\end{align}
where we used the form of $c(k_1, k_2)$ in Eq.~\eqref{eq:prod-sym-cost} to write
\begin{equation}
P_n(\xi_1, \xi_2) = \sum_{i=1}^2 \alpha_i \xi_i +  2\Gamma f(\xi_1, \xi_2) + O\Big(\frac{1}{n}\Big).
\end{equation}

Let us denote $H(\xi)= -\xi \log \xi - (1-\xi) \log (1-\xi)$ as the binary entropy function, then Stirling's approximation yields
\begin{align*}
    \binom{n_i}{k_i} = \frac{1}{\sqrt{2\pi k_i(1-k_i/n_i)}} \exp[n_i H(k_i/n_i)] \times (1+o(1)).
\end{align*}
We next approximate the sums in \eqref{eq:prod-sym-sum} with integrals using the Euler-Maclaurin formula and group the boundary terms in the error as we did previously in Section~\ref{sec:proof-unit-prob}, yielding
\begin{align}
\braket{\sv|\psi_\param} &= \frac{1}{2^n} \int_0^{n_1} dk_1 \int_0^{n_2} dk_2 \binom{n_1}{k_1} \binom{n_2}{k_2} \exp\Big[i n\frac{\pi}{2} P_n\Big(\frac{k_1}{n_1}, \frac{k_2}{n_2}\Big) \Big] \times(1+o(1)) \nonumber \\
&= \frac{1}{2^n}\int_0^{n_1} \int_0^{n_2}  \exp\Big[\sum_{i=1}^2 n_i H\big(\frac{k_i}{n_i}\big)+ i n\frac{\pi}{2} P_n\Big(\frac{k_1}{n_1}, \frac{k_2}{n_2}\Big) \Big] \prod_{i=1}^2 \frac{dk_i}{\sqrt{2\pi k_i (1-k_i/n_i)}}  \times (1+o(1)) \nonumber \\
&=\iint_0^{1}  \exp\Big[ n\sum_{i=1}^2 \alpha_i H(\xi_i) - n\log 2+ i n\frac{\pi}{2} P_n(\xiv) \Big] \prod_{i=1}^2 \frac{d\xi_i \sqrt{\alpha_i n} }{\sqrt{2\pi \xi_i (1-\xi_i)}} \times  (1+o(1)) \nonumber \\
&=\iint_0^{1}  \exp\Big[ n S(\xiv) + i\epsilon_n(\xiv) \Big] \prod_{i=1}^2 \frac{d\xi_i\sqrt{\alpha_i n}  }{\sqrt{2\pi \xi_i (1-\xi_i)}} \times (1+o(1)), \label{eq:integral-2d}
\end{align}
where we changed integration variable to $\xi_i=k_i/n_i$, regrouped terms to obtain
\begin{equation}
S(\xiv) = \sum_{i=1}^2 \Big[\alpha_i H(\xi_i) + i\frac{\pi}{2} \alpha_i \xi_i \Big] - \log 2 + i\pi \Gamma f(\xiv),
\end{equation}
and denoted $\epsilon_n(\xiv)$ as the sub-leading order terms in $P_n(\xv)$ that remains uniformly bounded as $n\to\infty$.

Observe that $\Re[S(\xiv)]$ achieves a unique maximum at $\xiv^* = (1/2,1/2)$ in the domain of integration $\xiv\in (0,1)^2$.
We want to apply the saddle point method to evaluate the integral in Eq.~\eqref{eq:integral-2d}, but we will need to ensure $\xiv^*$ is a non-degenerate saddle point of $S(\xiv)$.
This translates to the condition that $\nabla S(\xiv^*)=0$ and $\det H_S(\xiv^*)\neq 0$, where $H_S$ is the Hessian of $S$. To this end, let us examine $\nabla S(\xiv^*)$, whose components are
\[
\pder[S]{\xi_1}\Big|_{\xiv=\xiv^*} = \frac{i\pi}{2} \Big[\alpha_1 + \frac{\Gamma}{2^{\ell-2}} \sum_\mu \kappa_\mu (a_\mu-b_\mu) \Big],
\qquad
\pder[S]{\xi_2}\Big|_{\xiv=\xiv^*} = \frac{i\pi}{2} \Big[\alpha_2 + \frac{\Gamma}{2^{\ell-2}} \sum_\mu \kappa_\mu (c_\mu-d_\mu) \Big].
\]
Since the only free parameter is $\Gamma$, these components can be simultaneously made zero only if 
\begin{equation} 
\textstyle
\frac{1}{\alpha_1} \sum_\mu \kappa_\mu (a_\mu - b_\mu) = \frac{1}{\alpha_2} \sum_\mu \kappa_\mu (c_\mu - d_\mu) \neq 0,
\end{equation}
which is the condition in Eq.~\eqref{eq:grad-cond}.
Assuming this condition, we choose
\begin{equation} \label{eq:opt-Gam-prod}
\Gamma =  -2^{\ell-2} \alpha_1  \Big[\sum_\mu \kappa_\mu (a_\mu - b_\mu)\Big]^{-1}
\end{equation}
to ensure that $\nabla S(\xiv^*)=0$.
Let us next examine the second-order derivatives in the Hessian matrix elements, which are
\begin{align*}
\frac{\partial^2 S}{\partial \xi_1\xi_2}\Big|_{\xiv=\xiv^*}
&= \frac{i \pi\Gamma}{2^{\ell-2}} \sum_\mu \kappa_\mu (a_\mu - b_\mu)(c_\mu - d_\mu), \\
\frac{\partial^2 S}{\partial \xi_1^2}\Big|_{\xiv=\xiv^*}
&= - 4 \alpha_1 +  \frac{i \pi\Gamma}{2^{\ell-2}} \sum_\mu \kappa_\mu [(a_\mu-b_\mu)^2 - (a_\mu + b_\mu)], \\
\frac{\partial^2 S}{\partial \xi_2^2}\Big|_{\xiv=\xiv^*}
&= - 4 \alpha_2 +  \frac{i \pi\Gamma}{2^{\ell-2}} \sum_\mu \kappa_\mu [(c_\mu-d_\mu)^2 - (c_\mu + d_\mu)].
\end{align*}
Note that the condition that $\det H_S(\xiv^*)\neq0$ with the $\Gamma$ chosen in \eqref{eq:opt-Gam-prod} is equivalent to the condition \eqref{eq:hessian-cond} in Assumption~\ref{as:prod-sym-form}.
Hence, $\xiv^*$ is indeed a non-degenerate saddle point of $S(\xiv)$ under the assumption.
Now that we have met the requirements for applying the saddle point method, we can evaluate the integral in Eq.~\eqref{eq:integral-2d} to get
\begin{align}
\lim_{n\to\infty} |\braket{\sv|\psi_\param}|^2 = \lim_{n\to\infty} \Big(\frac{2\pi}{n} \frac{n}{2\pi/4}\Big)^2 \alpha_1\alpha_2 e^{2n \Re[S(\xiv^*)]}  \frac{1}{|\det H_S(\xiv^*)|} = \frac{16\alpha_1\alpha_2}{|\det H_S(\xiv^*)|},
\end{align}
where we used the fact that $\Re[S(\xiv^*)] = 0$.
\end{proof}

\section{QAOA can solve near-symmetric CSPs in polynomial time}
\label{sec:near-symmetric}

Building on the result on symmetric CSPs in the previous section, we now consider more general families of optimization problems where the symmetry is broken while still enabling a high success probability for the QAOA. One reason to break the symmetry is to construct problems that are potentially more challenging for classical algorithms, since they might not be able to take advantage of the symmetry if known in advance.

While there are multiple ways to break the assumption of symmetry, here we focus on an explicit method to construct CSPs with broken symmetry that uses only local clauses. The key idea is to take any symmetric CSPs consisting of local clauses and ``sparsify'' it by randomly sampling the clauses that make up the cost function. (We will briefly discuss another method to break symmetry in Remark~\ref{rem:break-sym} at the end of this section.)
Our result shows if the QAOA success probability is high on the original symmetric CSP, then it remains high on the sparsified version with the symmetry broken.

The starting point of our construction is as follows. Suppose we are given a symmetric CSP of the form $C(\xv) = \sum_\alpha C_\alpha(\xv)$, where each $C_\alpha$ is a single unit-weight constraint clause. For example, the clause $\phi_\alpha = x_i \wedge x_j \wedge \lnot x_k$ would correspond to $C_\alpha(\xv) =(1-x_i)(1-x_j)x_k$.
We then consider a ``sparsified'' version of the same CSP
\begin{equation}
	\tilde{C}_w(\xv) = \sum_{\alpha} w_\alpha C_\alpha(\xv),
\end{equation}
where $w_\alpha$'s are  i.i.d. random weights drawn from some distribution.
For concreteness, we consider a specific case where each $w_\alpha$ is a Bernoulli random variable with mean $f$.

Note that the symmetry is typically broken in the sparsified CSP.
Nevertheless, we show in the below theorem that as long as the QAOA succeeds for the original symmetric CSP with sufficiently small parameter $\gamma$, the QAOA also succeeds for the near-symmetric, sparsified CSPs.

\begin{thm}\label{thm:sparsified}
Assume the 1-step QAOA with $(\gamma,\beta)= (\gamma_0/n^{\ell-1}, -\pi/4)$ for some constant $\gamma_0$ achieves $\Omega(1)$ ground state probability on a symmetric $\ell$-CSP with $\Theta(n^\ell)$ unit-weight constraint clauses. Consider a random Bernoulli-sparsified version of the same CSP where we choose $f=d_n/n^{\ell-1}$ so that the interaction graph has average degree $d_n$. Then as long as $d_n=\Omega(n)$, the QAOA can also achieve $\Omega(1)$ ground state probability on the sparsified CSP with $\Omega(1)$ probability over the random instances.
\end{thm}

As shown previously in Theorems~\ref{thm:unit-prob} and \ref{thm:prod-sym}, we can achieve $\Omega(1)$ ground state probability for many families of symmetric $\ell$-CSPs using the 1-step QAOA with $\gamma = \Theta( 1/n^{\ell-1})$ .
Combined with the above theorem, this tells us that we can also solve the corresponding near-symmetric CSPs with high probability.

To prove Theorem~\ref{thm:sparsified}, we start by showing that in the small angle limit, the 1-step QAOA state generated by the sparsified $\tilde{C}_w$ with appropriately rescaled $\gamma$ is approximately the same as the QAOA state generated with the original $C$. (Note that $C$ needs not be symmetric for this lemma.)

\begin{lem}
	\label{lem:sparsified}
	Given any CSP $C=\sum_\alpha C_\alpha$ consisting of unit-weight constraints, i.e. $C_\alpha^2 = C_\alpha$, let $\tilde{C}_w$ be its sparsified version with randomly chosen i.i.d. weights $(w_\alpha)_\alpha$ drawn from a distribution with mean $\E[w_\alpha]=f\neq 0$ and variance $\sigma^2$. Then
	\begin{equation}
		\E_w \left\| \big( e^{i\gamma C} - e^{i(\gamma/f) \tilde{C}_w} \big) \ket{+}^{\otimes n}\right\|^2
		\le \frac{\gamma^2 \sigma^2 }{f^2} \, \E_{\xv} C(\xv)
	\end{equation}
\end{lem}
\begin{proof}
Let $\ket{\delta\psi} = (e^{i\gamma C} - e^{i(\gamma/f) \tilde{C}_w} ) \ket{+}^{\otimes n} $, then
\begin{align*}
	\ket{\delta\psi} &=  \frac{1}{2^{n/2}} \sum_{\xv} \Big( e^{i\gamma C(\xv)} - e^{i(\gamma/f) \tilde{C}_w(\xv)} \Big) \ket{\xv}, \\
	\text{and} \quad \|\ket{\delta\psi}\|^2 &= \E_{\xv} \left|e^{i\gamma C(\xv)} - e^{i(\gamma/f) \tilde{C}_w(\xv)}\right|^2.
\end{align*}
For any bit string $\xv$, let
\begin{equation*}
	\Delta(\xv) := C(\xv) - \frac{1}{f}\tilde{C}_w(\xv) = \frac{1}{f}\sum_\alpha ( f-w_\alpha  ) C_\alpha(\xv).
\end{equation*}
It's easy to see that $\E_w[\Delta(\xv)]=0$. Furthermore,
\begin{align*}
	\E_w [\Delta^2(\xv)] = \frac{1}{f^2}  &\sum_\alpha \sum_{\alpha'} \E_w\Big[(f - w_\alpha)(f-w_{\alpha'})\Big]  C_\alpha(\xv) C_{\alpha'}(\xv).
\end{align*}
Noting that the covariance between $w_\alpha$ and $w_{\alpha'}$ is zero due to independence if $\alpha\neq \alpha'$, we have
\begin{align} \label{eq:sparsified-variance}
	\E_w [\Delta^2(\xv)]
	= \frac{1}{f^2} \sum_\alpha \E_w[(f-w_\alpha)^2] C_\alpha^2(\xv) = \frac{\sigma^2}{f^2} \sum_\alpha C_\alpha(\xv) = \frac{\sigma^2}{f^2} C(\xv),
\end{align}
where we used the assumption of unit-weight constraints to set $C^2_\alpha = C_\alpha$.
Finally, using the fact that $|e^{ia}-e^{ib}| \le |a-b|$ for any $a,b\in\R$, we have 
\begin{equation*}
	\E_w \|\ket{\delta\psi}\|^2 \le \E_{w,\xv} [\gamma^2 \Delta^2 (\xv)] = \frac{\gamma^2 \sigma^2}{f^2} \E_{\xv} C(\xv),
\end{equation*}
concluding the proof.
\end{proof}

\begin{proof}[Proof of Theorem~\ref{thm:sparsified}]
Here we focus on the case when $C$ is a symmetric $\ell$-CSP with $\Theta(n^\ell)$ unit-weight clauses, and each $w_\alpha$ is a Bernoulli random variable. Then $\E_\xv C(\xv) = \Theta(n^\ell)$ and $\sigma^2 = f(1-f)$. 
Let $\ket{\psi}$ and $\ket{\tilde{\psi}_w}$ be the QAOA states corresponding to original symmetric CSP and the sparsified CSP, respectively. Then using Lemma~\ref{lem:sparsified} and Markov's inequality, we have for any $K>1$,
\[
\|\ket{\psi}-\ket{\tilde{\psi}_w}\| \le O(K \gamma n^{\ell/2}\sqrt{(1-f)/f}) = O(K\sqrt{n/d_n})
\]
with $1-1/K^2$ probability over the weights $w$. From the  assumption we have $|\braket{\psi|\sv}|^2\ge \Omega(1)$, which together implies
\begin{align*}
	\big|\braket{\tilde{\psi}_w|\sv}\big|^2 \ge \Omega(1) -  O(K\sqrt{n/d_n}).
\end{align*}
	This remains at least $\Omega(1)$ if $d_n = \Omega(K^2 n)$.
\end{proof}

\begin{rem}
We now provide some comments on the proof of Theorem~\ref{thm:sparsified}.
In the parameter regime $\gamma=\Theta(1/n^{\ell-1})$ considered, both $\gamma C = \Theta(n)$ and $\gamma \tilde{C}_w/f=\Theta(n)$ are not small. However, using Lemma~\ref{lem:sparsified} we find that $\big\|(e^{i\gamma C}-e^{i\gamma\tilde{C}_w/f})\ket{+}^{\otimes n}\big\|$ is small.
To give some intuition for why this is the case, let us think of $\tilde{C}_w/f$ as a rescaled version of the sparsified cost function so that on average,
$\E_w[\tilde{C}_w (\xv) /f] = C(\xv)=\Theta(n^\ell)$.
Reading off Eq.~\eqref{eq:sparsified-variance} and plugging in $\sigma^2=f(1-f)$ for Bernoulli random weights, we see that its standard deviation is $\sqrt{(1-f)/f} \sqrt{C(\xv)} = O(n^{\ell-1})$, where the last equality used the condition $f=\Omega(1/n^{\ell-2})$ from the theorem statement.
Consequently, $\gamma \tilde{C}_w(\xv)/f=\gamma C(\xv) \pm O(1)$, which means the fluctuations due to sparsification is relatively small compared to the typical $\Theta(n)$ phase in the operator $e^{i\gamma \tilde{C}_w/f}$, resulting in approximately the same quantum state as the original QAOA.
\end{rem}

\begin{rem}[Learning the hidden string from observing individual clauses]
\label{rem:LPN}
In the setting of near-symmetric CSPs constructed with the above sparsification method, one natural question to ask is: Can one exploit the knowledge of the approximate symmetry to learn the hidden string $\sv$ by observing the sampled clauses in the cost function and running a classical algorithm?
As an example, let us consider a near-$S_n$-symmetric CSP constructed by symmetrizing a local clause function and then sparsifying. 
Specifically, for any bit string $\av\in \{0,1\}^n$ with Hamming weight $\ell$, consider the $\ell$-local clause
\begin{equation*}
C_\av(\xv) = \av\cdot (\xv\oplus \sv),
\end{equation*}
where $\av\cdot\bv \in \{0,1\}$ denotes the modulo-2 inner product of $\av$ and $\bv$.
The cost function of the $S_n$-symmetric CSP can be written as $C(\xv)= \mathcal{N}\sum_{\pi \in S_n} C_{\pi(\av)}(\xv)$, where $\mathcal{N}$ is a normalization constant accounting for repeated clauses. 
Now, in a sparsified version of the cost function, observing a clause $C_\bv$ allows one to infer the value of $e_\bv=\bv\cdot \sv$. With $\Omega(n)$ observed clauses, one can solve the corresponding system of linear equations with Gaussian elimination to determine $\sv$ in this case.

However, there is a strong limitation to this approach of recovering $\sv$. Note that the problem of learning $\sv$ from the observations $(\av, e_\av=\av\cdot \sv)$ is a parity learning problem (see e.g., \cite{Blum2003LPN}).
This is closely related to the learning parity with noise (LPN) problem, which is believed to be so difficult that it has been used as a cryptographic assumption~\cite{Pietrzak2012LPNcrypto}.
In the above setting, we have limited ourselves to $\av$'s that have bounded Hamming weights, which relates to the sparse LPN problem (see e.g., \cite{Applebaum2010public, chen2024sparselpn}).
Indeed, we can simulate the sparse LPN problem by including in the cost function symmetrized clauses of the form $C'_\av(\xv) = 1-C_\av(\xv)$, sampled with a different probability than $C_\av$.
Given a random clause, it is impossible to tell whether it came from $C_\av$ or $C'_\av$, and observing either would yield the equation $\av\cdot\sv = e_\av$ or $\av\cdot\sv = 1-e_\av$, respectively.
This results in a noisy system of linear equations equivalent to the sparse LPN problem, which is believed to take exponential time to solve classically in various regimes \cite{Applebaum2010public,chen2024sparselpn}.
\end{rem}

\begin{rem}[Other methods to break symmetry]
\label{rem:break-sym}
Another way to generalize our results to symmetry-broken problems while retaining the QAOA's success is as follows. Suppose the QAOA achieves $\Omega(p(n))$ ground state overlap with the solution bit string for CSPs that are $G$-symmetric for some function $p(n)$, then it is not difficult to see that the QAOA can also solve any CSP whose cost function $C(\xv)$ obeys $G$-symmetry for all but $o(p(n))$ fraction of all $2^n$ bit strings inputs. One method to explicitly construct such a symmetry-broken CSP is to take a $G$-symmetric CSP, randomly choose $o(p(n))$ fraction of $2^n$ bit strings and alter their cost function value. However, this will result in a highly nonlocal CSP that may not resemble any natural or practically relevant optimization problems.
\end{rem}

\section{Performance of classical algorithms}
\label{sec:classical}

\subsection{Oracular setting}
We first consider classical algorithms in the oracular setting where the algorithm can only query the cost function as a black box.
In this setting, we show that there {\em exists} an efficient classical algorithm to solve any $S_n$-symmetric CSPs satisfying Assumption~\ref{as:distinct-values} with $O(n/\log n)$ queries, and that this is tight.
We also show that a natural class of hill-climbing algorithms is unable to solve them all efficiently.

\begin{claim} \label{cl:alg-for-distinct}
For any $S_n$-symmetric CSP based on a known cost function $C(\xv)=c(|\xv\oplus\sv|)$ taking distinct values as in Assumption~\ref{as:distinct-values}, there is a classical algorithm that uses $O(n/\log n)$ queries to $C$ to find the solution, and runs in polynomial time. 
Furthermore, any classical algorithm needs $\Omega(n/\log n)$ queries to find the solution even knowing the shape of $c(k)$ in advance.
\end{claim}

\begin{proof}
As the algorithm knows the cost function $C$ is $S_n$-symmetric and it takes distinct values, we can assume that $C(\xv) = |\xv\oplus\sv|$, by remapping the output of the cost function if necessary. We first describe a (well-known) simple classical algorithm which uses $O(n)$ queries and makes $O(n)$ additional operations. Query all bit strings of Hamming weight at most 1. The query to $0^n$ returns $|\sv|$, while each query to a bit string of Hamming weight 1 returns either $|\sv|-1$ or $|\sv|+1$, depending on whether the corresponding bit of $\sv$ is equal to 0 or 1, enabling us to output $\sv$.

Perhaps surprisingly, this complexity can be improved to $O(n/\log n)$ queries. First we query $\xv=0^n$, as before, to learn the Hamming weight of $\sv$. Subsequent queries to bit strings $\xv$ return the number of bits that differ between $\sv$ and $\xv$. Using this and our knowledge of the Hamming weights of $\sv$ and $\xv$ allows us to compute $|\sv \wedge \xv|$. This is equivalent to learning the number of 1's in $\sv$, restricted to an arbitrary subset $X$ of the bits. The problem of learning a bit string $\sv$ given queries of this form is known as quantitative group testing (QGT), and it turns out that we can learn an arbitrary bit string with $O(n/\log n)$ subset queries, and in polynomial time~\cite{wang16,gebhard22,shipra24}.

Finally, we provide a matching lower bound on the number of queries used by any classical algorithm. We note that to find the planted string $\sv$ requires $n$ bits of information. On the other hand, each classical query provides at most $O(\log(n))$ bits of information from the returned value of the cost function, since there can be at most $n+1$ distinct values. Hence, $\Omega(n/\log n)$ queries are necessary  to recover $\sv$.
\end{proof}

\begin{rem}[$O(1)$ quantum vs. $\Omega(n/\log n)$ classical query]
Recall that we have shown in Theorem~\ref{thm:distinct-values} that the QAOA with 1 query achieves $\Omega(1/\sqrt{n})$ of recovering the bit string as long as the cost function takes on distinct values, and thus $O(\sqrt{n})$ quantum queries is enough to find $\sv$ with high probability.
For specific cases such as the cost in Eq.~\eqref{eq:example-3SAT}, 1-step QAOA can find $\sv$ with $\Omega(1)$ probability due to Theorem~\ref{thm:distinct-values} and Lemma~\ref{lem:distinct-values}.
Hence, in these examples, $O(1)$ quantum queries are sufficient to recover $\sv$ compared to the $\Omega(n/\log n)$ classical query lower bound.
\end{rem}

\begin{rem}
We remark that the algorithms in Claim~\ref{cl:alg-for-distinct} heavily rely on the fact that the cost function takes on $n+1$ distinct values. 
For the first algorithm, one can for example adversarially construct a cost function where $c(|s|-1)=c(|s|+1)$, so that queries with bit strings of Hamming weight 1 would give no information about the bits of $\sv$. Note this adversarial cost function still satisfies Assumption~\ref{as:approximate-distinct-values} when $\big||\sv|- n/2\big| \ge \Omega(\sqrt{n\log n})$, allowing the QAOA to succeed.
Furthermore, the second algorithm based on QGT can run into issues when the evaluations of $|\sv \wedge \xv|$ are noisy, which can manifest from either non-distinctness or sparsification.
In some cases, for sufficiently strong noise, a superpolynomial lower bound on the query complexity of the noisy QGT problem has been shown using information-theoretical arguments~\cite{Chen2017QGTnoise}.
\end{rem}

In this oracular setting, we also consider a {\em hill climbing} algorithm, which is one that always makes (deterministic or random) local moves that reduce the cost function.

\begin{claim}
	There is an oracular $S_n$-symmetric CSP based on a cost function $c(|\xv|)$ taking distinct values, such that any hill climbing algorithm $\mathcal{A}$ using a random initial starting point requires exponentially many starting points to find the solution.
\end{claim}

\begin{proof}
Consider the cost function $c(0) = 0$, $c(k) = n-k+1$ for $k=1,\dots,n$. If $\mathcal{A}$ does not happen to choose $\xv=0^n$ as the initial bit string, all subsequent moves will increase the Hamming weight, eventually ending up in the local minimum $\xv = 1^n$. So $\Omega(2^n)$ random initial guesses are needed to find the unique solution.
\end{proof}

This does not contradict the previous result because the algorithms in Claim~\ref{cl:alg-for-distinct} are not hill climbing algorithms, in terms of the original cost function (as opposed to the cost function after remapping of the costs to $C(\xv) = |\xv\oplus\sv|$, in which case one can define a hill-climbing algorithm which does work).

\subsection{Simulated annealing}
Simulated annealing is a powerful general-purpose classical optimization algorithm that often serves as a point of comparison for quantum algorithms.
Indeed, previous work on quantum annealing~\cite{farhi2002QAAvsSA, muthukrishnan2016} for $S_n$-symmetric problems has compared its performance to simulated annealing. Their results have shown that quantum annealing is generally superior to simulated annealing, although there are problems where both algorithms take exponential time to find the global minimum.

For the problems we consider this paper, it is generally not difficult to construct instances where simulated annealing will take exponential time to find the solution.
Take for example the cost function in Eq.~\eqref{eq:example-3SAT-cost}, where there is a global minimum at $k=0$ and false minimum near $k=n$, and the cost decreases with $k$ in the range $k/n\ge 1/3$. Simulated annealing with a random initial string will start with $k\approx n/2$ with high probability. 
While there is a possibility of updates that lower $k$, it requires a sequence of $n/2-n/3=n/6$ bit flips to climb over the hill which has an energy barrier of height $\Theta(n^2)$. At any low temperature $T=O(n)$, each energy increasing move is only accepted with probability $\exp(-\Delta E/T) = \exp(-\Omega(1))$ which is uniformly bounded away from 1. Hence, the transition probability to escape the false minimum with single-bit updates is exponentially small in $n$, implying that simulated annealing will take exponential time.

\subsection{Best solvers from SAT and Max-SAT competitions}
\label{sec:benchmark}

In this section, we move beyond the limited examples of classical algorithms discussed above and benchmark the performance of more general, state-of-the-art classical algorithms on the symmetric and near-symmetric CSPs.
For this benchmark, we will construct a few families of $S_n$-symmetric and $(S_{n/2})^2$-symmetric SAT problems in the form of Boolean formulas.
We also construct related families of near-symmetric CSPs by using the random sampling trick discussed in Section~\ref{sec:near-symmetric}.
These constructed problems are then put to the test by six state-of-the-art classical algorithms, including some front-runners in recent SAT and Max-SAT competitions. While we find that some of these problems can be solved by a subset of the classical algorithms in polynomial time, we also discover families of problems with an approximate $(S_{n/2})^2$ symmetry that appear to require exponential run time from all tested classical solvers.
Since we have shown that 1-step QAOA can succeed in solving these problems in polynomial time, these constructed local CSPs serve as somewhat convincing candidates for exponential quantum speedup by a low-depth QAOA in optimization.

\paragraph{Construction of Benchmark Problems.}
We begin by explaining our constructions of $S_n$-symmetric SAT problems that will be used in the benchmarking test.
Given a solution string $\sv$ to be planted, we first transform the input $\xv$ into $\yv=\xv\oplus \sv$ and then construct a Boolean formula based on the bits of $\yv$.
For convenience of notation, we will denote $\ell$N$m$ for any $\ell\ge m\ge 0$ as the Boolean formula made by considering all $\ell$-subsets of $n$ bits and including a clause containing a disjunction of the selected literals with $m$ negations. Furthermore, we include all permutations of each subset whenever $m\neq 0,\ell$.
For example,
\begin{align*}
2\tN2 = \bigwedge_{i< j} (\lnot y_i \vee \lnot y_j ),
\qquad
3\tN1 = \bigwedge_{i\neq j \neq k} (\lnot y_i \vee y_j \vee  y_k),
\qquad
4\tN2 = \bigwedge_{i\neq j \neq k \neq l} (\lnot y_i \vee \lnot y_j \vee  y_k \vee y_l).
\end{align*}
As a shorthand, we also denote $\ell\tN m_1m_2 = \ell\tN m_1 \wedge \ell\tN m_2$.
For any given family of symmetric SAT problems, we may also construct a near-symmetric SAT problem by choosing to keep each Boolean clause with probability $f$.

As an example, consider the Max-4-SAT problem $4\tN1 \wedge 3\tN 31 \wedge 2\tN0$ whose cost function, when written in terms of $k=|\yv|=|\xv\oplus \sv|$, is
\begin{equation*}
c(k) = k(n-k)(n-k-1)(n-k-2) + k(n-k)(n-k-1) + \frac{k(k-1)(k-2)}{3!} + \frac{(n-k)(n-k-1)}{2!}.
\end{equation*}
Note that this example Boolean formula is not satisfiable, and there are two (frustrated) local minima at $k=0$ and near $k=n$, respectively. Note that the $k=0$ local minimum acquires $\Theta(n^2)$ energy from the 2N0 subformula, whereas the $k=n$ local minimum is penalized with $\Theta(n^3)$ energy by the 3N3 subformula. Hence, this construction ensures that $k=0$ is the true global minimum that corresponds to the exact planted solution.

We also consider $(S_{n/2})^2$-symmetric CSPs which allow us to construct more complex problems with additional local minima. In this construction, we first divide the $n$ bits into two equal-sized groups. Then there are two types of local clauses we can use to construct an $(S_{n/2})^2$-symmetric Boolean formula: (1) clauses that contains bits from one group but not the other, and (2) clauses that contain bits from both groups. We can construct formula from the first type in the same way as the $S_n$-symmetric case above, and we will denote these formulas in the same way as $\ell\tN m$.
For the second type of clauses, we consider all $\ell_1$-subsets from the first $n/2$ bits and all $\ell_2$-subsets from the second $n/2$ bits, and construct a clause with $m_1$ negations in the $\ell_1$-subset and $m_2$ negations in the $\ell_2$-subset. As before, we consider all permutations within the $\ell_j$-subset unless $m_j=0$ or $m_j=\ell_j$.
For brevity, we denote the Boolean formula constructed with this method as $(\ell_1, \ell_2)\tN (m_1, m_2)$.
Moreover, to satisfy Assumption~\ref{as:prod-sym-form} that leads to a guarantee for QAOA's success in solving the problem via Theorem~\ref{thm:prod-sym}, we further symmetrize the formula constructed from both types of clauses by joining another copy of the formula where the roles of the two groups are switched.

For concreteness, let us illustrate this construction with two example $(S_{n/2})^2$-symmetric formulas: $5\tN1$ and $(3,2)\tN(1,0)$. Given the planted string $\sv$, we set $\yv=\xv\oplus \sv$ and partition $\yv$ into two equal-sized substrings $(\yv^\G1, \yv^\G2)$. Then,
\begin{align*}
5\tN1 &= \bigwedge_{i\neq j \neq k \neq l \neq m }(\lnot y_i^\G1 \vee y_j^\G1 \vee y_k^\G1 \vee y_l^\G1 \vee y_m^\G1)\bigwedge_{i\neq j \neq k \neq l \neq m } (\lnot y_i^\G2 \vee y_j^\G2 \vee y_k^\G2 \vee y_l^\G2 \vee y_m^\G2), \\
(3,2)\tN(1,0) &= \bigwedge_{i\neq j \neq k, l<m } (\lnot y_i^\G1 \vee y_j^\G1 \vee y_k^\G1 \vee y_l^\G2 \vee y_m^\G2) \bigwedge_{i\neq j \neq k, l<m } (\lnot y_i^\G2 \vee y_j^\G2 \vee y_k^\G2 \vee y_l^\G1 \vee y_m^\G1).
\end{align*}
Furthermore, denoting $k_1=|\yv^\G1|$ and $k_2=|\yv^\G2|$, the corresponding cost functions to leading order in $n$ are:
\begin{eqnarray*}
c_{5\tN1}(k_1, k_2) &=& k_1(n_1-k_1)^4 + k_2(n_2 - k_2)^4 + O(n^4), \\
c_{(3,2)\tN(1,0)}(k_1,k_2) &=& (k_1+k_2) (n_1-k_1)^2 (n_2 - k_2)^2 + O(n^4).
\end{eqnarray*}
Note these functions have four local minima, corresponding to $(k_1,k_2)=(0,0),(n_1,0),(0,n_2),(n_1,n_2)$.

\begin{figure}[b!]
\includegraphics[width=0.206\linewidth,valign=t]{{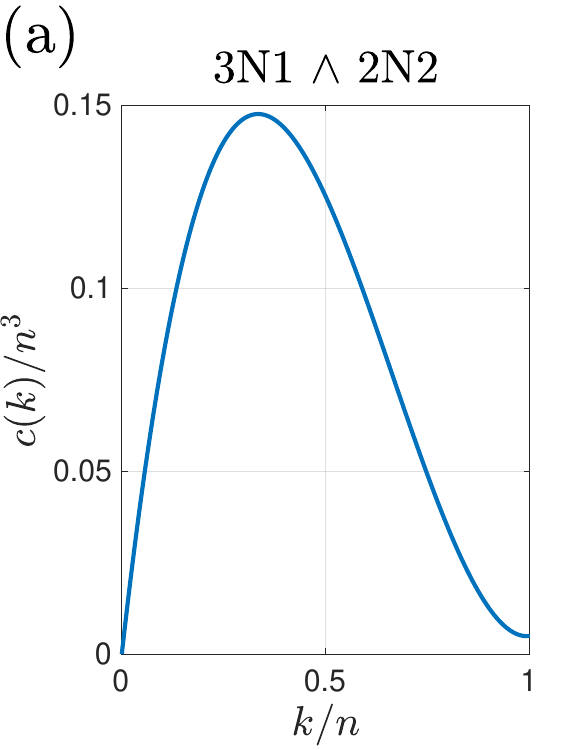}}
\includegraphics[width=0.397\linewidth,valign=t]{{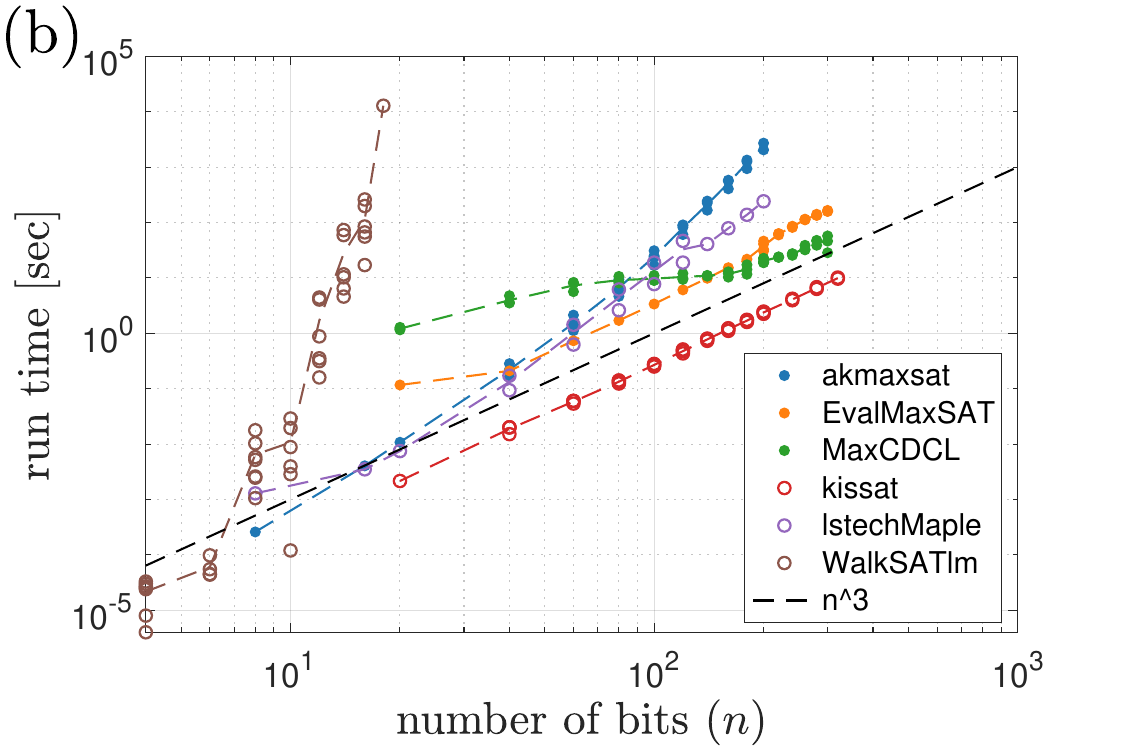}}
\includegraphics[width=0.397\linewidth,valign=t]{{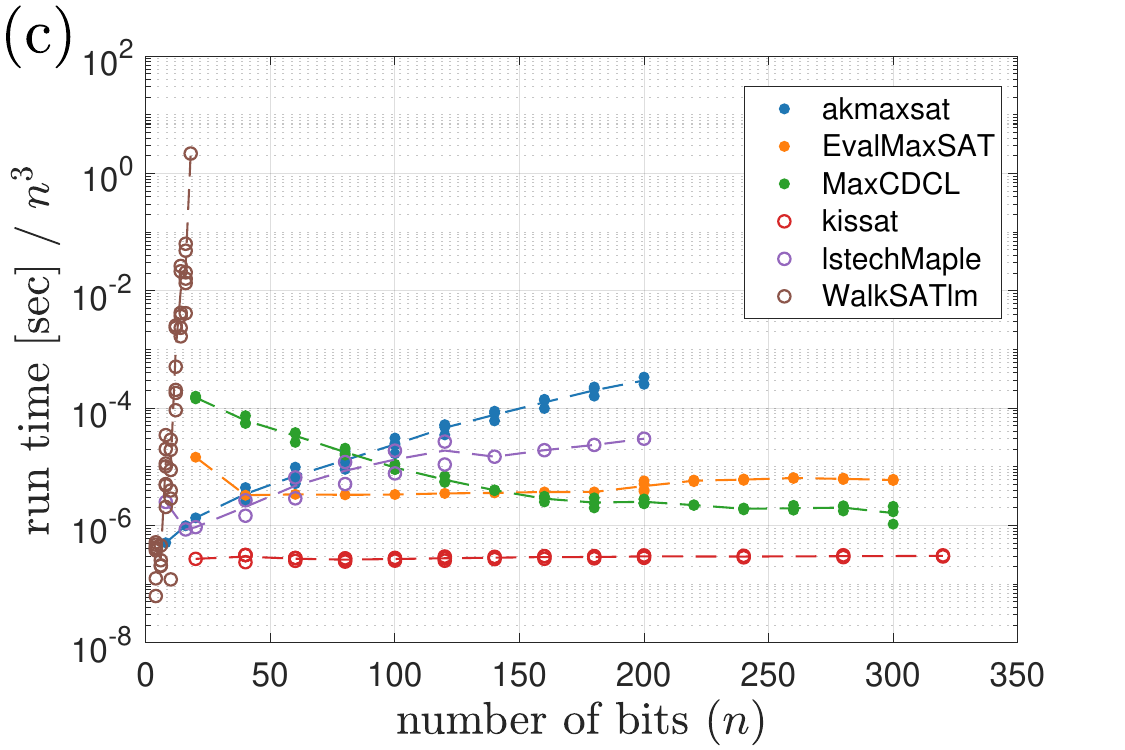}}
\caption{Performance of various classical SAT and Max-SAT solvers on a family of $S_n$-symmetric 3-SAT problems.
(a) Example cost function plotted versus normalized Hamming distance to the hidden string $\sv$ at $n=200$.
(b) Run times plotted as a function of number of bits $n$ on a log-log scale, where $n^3$ is plotted as a guide to the eyes. Individual dots correspond to runs with different choices of $\sv$, and dashed line connects the averages at each $n$.
(c) Run times divided by $n^3$, which is roughly the number of clauses in the formula, plotted on a log-linear scale. The flat lines of \alg{kissat}, \alg{EvalMaxSAT} and \alg{MaxCDCL} indicate their run time scales linearly with the problem size $\Theta(n^3)$.
\label{fig:benchmark-SAT3}
}
\end{figure}

\paragraph{Numerical Benchmark Results.}

For our numerical benchmarks, we have selected the following six state-of-the-art SAT and Max-SAT solvers:
\begin{enumerate}
\item \alg{WalkSATlm} (SAT solver): An improved version of the famous \alg{WalkSAT} algorithm. This was found to be best performant classical algorithm for random $\ell$-SAT in \cite{boulebnane2024SAT}.
\item \alg{kissat} (SAT solver): 2nd in Main Track UNSAT of SAT Competition 2021.
\item \texttt{lstechMaple} (SAT solver): 2nd in Main Track SAT in SAT Competition 2021.
\item \texttt{akmaxsat} (Max-SAT solver): Winner of MaxSAT Evaluation 2010, 2011, and 2012.
\item \texttt{EvalMaxSAT} (Max-SAT solver): Winner of Unweighted Exact Track in MaxSAT Evaluation 2023.
\item \texttt{MaxCDCL} (Max-SAT solver): Winner of Weighted Exact Track in MaxSAT Evaluation 2023.
\end{enumerate}

\begin{figure}[b!]
\includegraphics[width=0.5\linewidth]{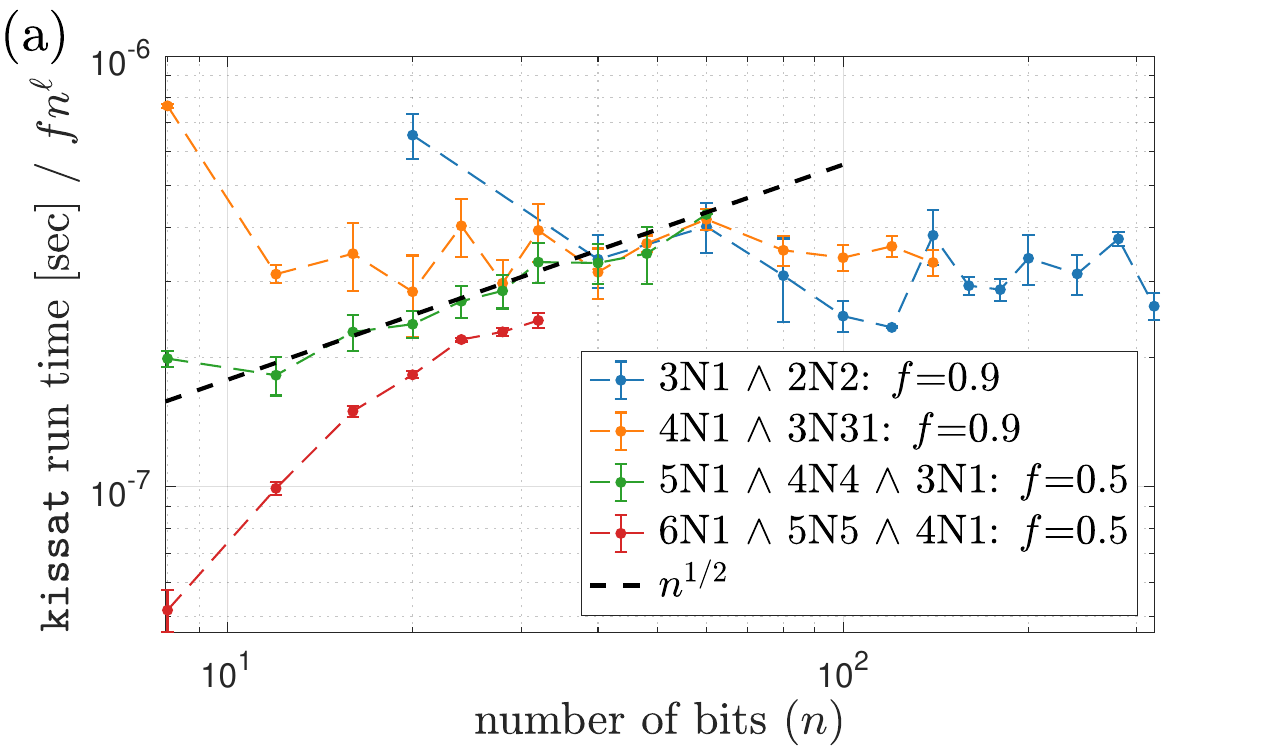}
\includegraphics[width=0.5\linewidth]{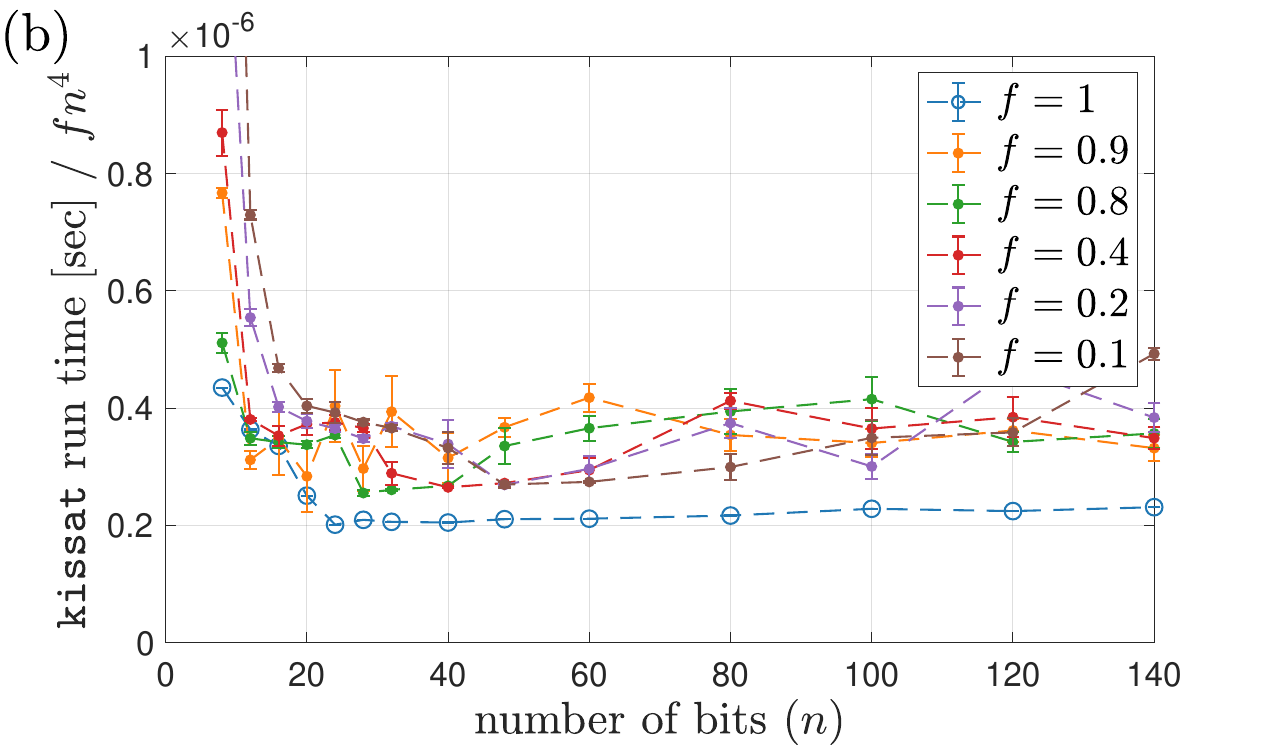}
\caption{\alg{kissat} run time rescaled by problem size for different near-$S_n$-symmetric $\ell$-SAT problems, where we explore the dependence with respect to locality $\ell$ in (a) and the sparsification fraction $f$ in (b). Here $f$ is the fraction of clauses sampled from the corresponding symmetric $\ell$-SAT problems to construct the near-symmetric problem. Dashed lines connect averages over 3$\sim$5 instances with different choices of $\sv$ and clause sampling realizations. Error bars are standard errors of the mean.
\label{fig:kissat-locality}}
\end{figure}

\begin{figure}[htb]
\includegraphics[width=0.206\linewidth,valign=t]{{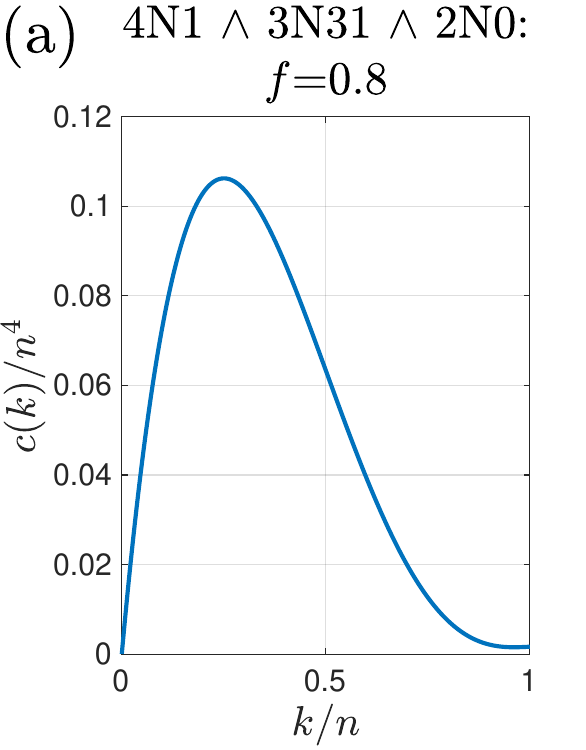}}
\includegraphics[width=0.397\linewidth,valign=t]{{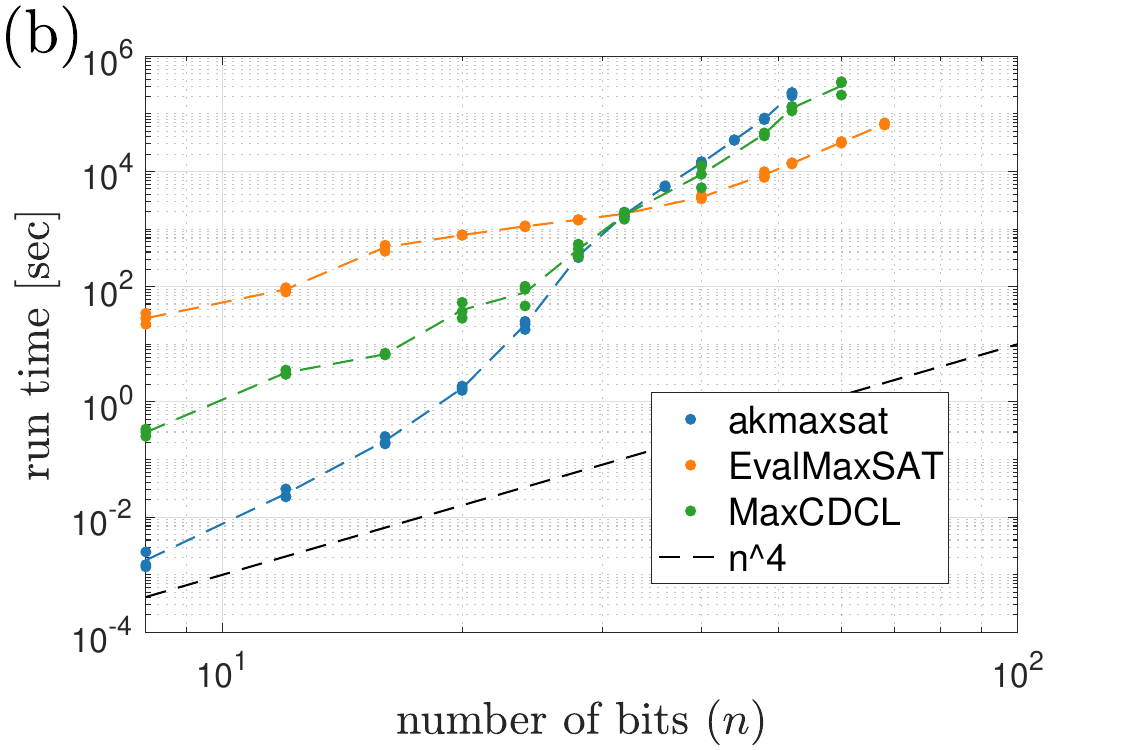}}
\includegraphics[width=0.397\linewidth,valign=t]{{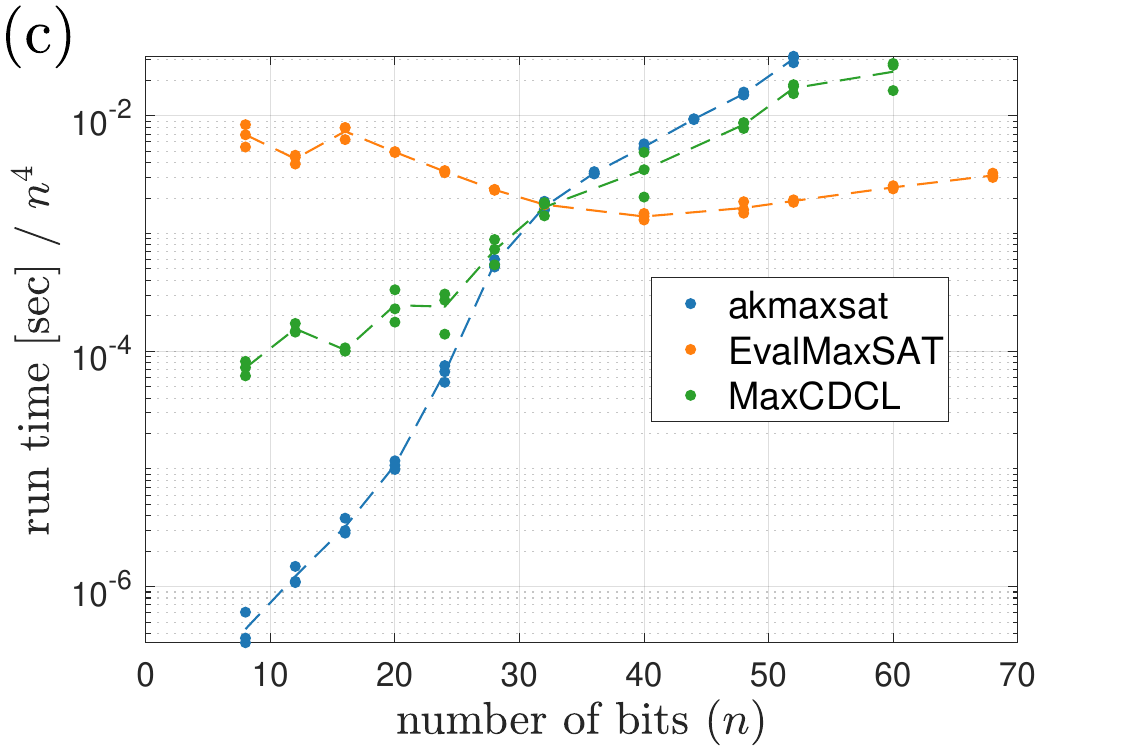}}\\
\includegraphics[width=0.206\linewidth,valign=t]{{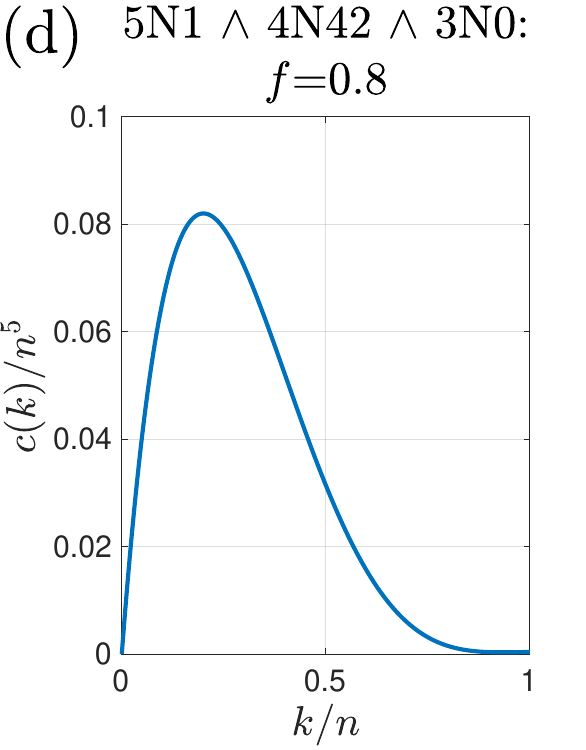}}
\includegraphics[width=0.397\linewidth,valign=t]{{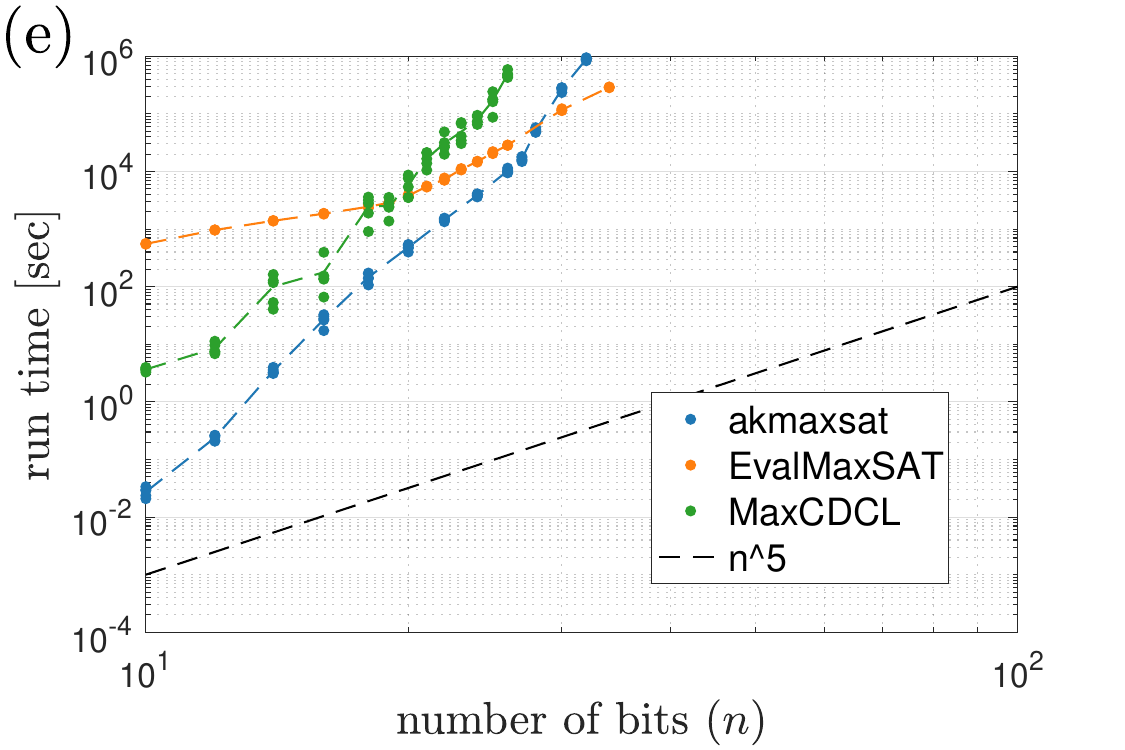}}
\includegraphics[width=0.397\linewidth,valign=t]{{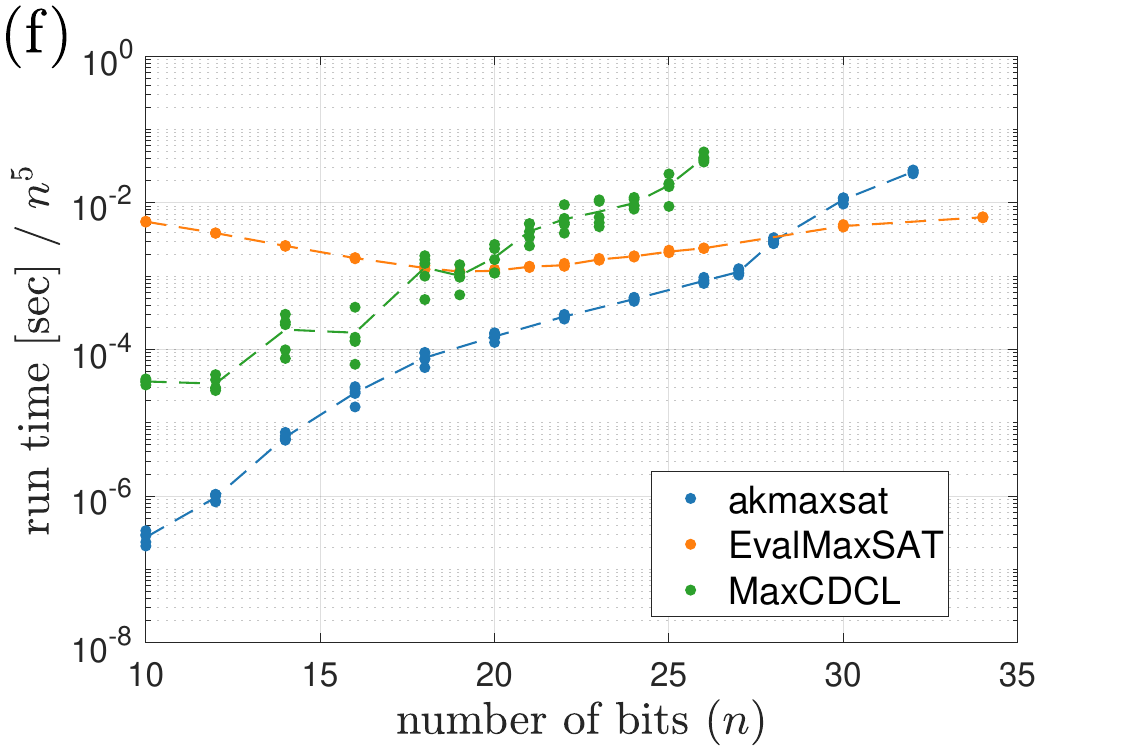}} 
\caption{Performance of classical Max-SAT solvers on two families of near-$S_n$-symmetric Max-$\ell$-SAT problems for $\ell=4$ (top row) and $\ell=5$ (bottom row), constructed with clause sampling.
(a)(d) Plot of the energy landscape as a function of the normalized Hamming distance to $\sv$.
(b)(e) Run times plotted as a function of number of bits $n$ on a log-log scale.
(c)(f) Run times divided by $n^\ell$ plotted as a function of $n$ on a log-linear scale.
Individual dots correspond to 3 problem instances at each $n$ with different clause sampling and choices of $\sv$, and dashed lines connect their averages.
The results suggest that \texttt{akmaxsat} and \texttt{MaxCDCL} take exponential time on both problems, but \texttt{EvalMaxSAT} seems to run in $O(n^\ell)$ time.
\label{fig:benchmark-SAT4}
}
\end{figure}

In Figure~\ref{fig:benchmark-SAT3}, we show our benchmark result on the symmetric 3-SAT problem $3\tN1 \wedge 2\tN2$. We observe that \alg{WalkSATlm}, the algorithm previously studied in Ref.~\cite{boulebnane2024SAT} as the best classical solver for random $\ell$-SAT and a key competitor against QAOA, turns out to be the worst out of the six algorithms studied here and runs in exponential time. The run time scaling for the other algorithms are milder. The best classical SAT solver for this problem is \alg{kissat} which seems to run in time linear in the number of clauses $\Theta(n^3)$. Similar linear time scaling is also apparently achieved by Max-SAT solvers \alg{EvalMaxSAT} and \alg{MaxCDCL}, as shown by the lines flattening as $n\to\infty$ in Figure~\ref{fig:benchmark-SAT3}(c).

Since \alg{kissat} appears to be the best solvers for symmetric SAT problems, we benchmark its performance on more complicated problems with higher locality and broken symmetry in Figure~\ref{fig:kissat-locality}. The result in subplot (a) indicates that for $\ell \le 4$, \alg{kissat} can solve these near-symmetric $\ell$-SAT in time linear in problem size $\Theta(n^\ell)$. However, for 5-SAT, the \alg{kissat} appears to have an extra $\Theta(\sqrt{n})$ factor overhead in its run time. Since the QAOA can solve these $\ell$-SAT problems with $O(n^\ell)$ gates, there appears to be an $\Omega(\sqrt{n})$ quantum speedup for $\ell\ge 5$.
In Figure~\ref{fig:kissat-locality}(b), we explore the effect of varying the fraction $f$ of clauses kept in the sparsification process. We find that sparsification appears to generally make the problem harder for \alg{kissat}, but does not seem to alter the scaling behavior of run times relative to problem size beyond a constant factor.

\begin{figure}[htb]
\vspace{10pt}
\includegraphics[width=0.206\linewidth,valign=t]{{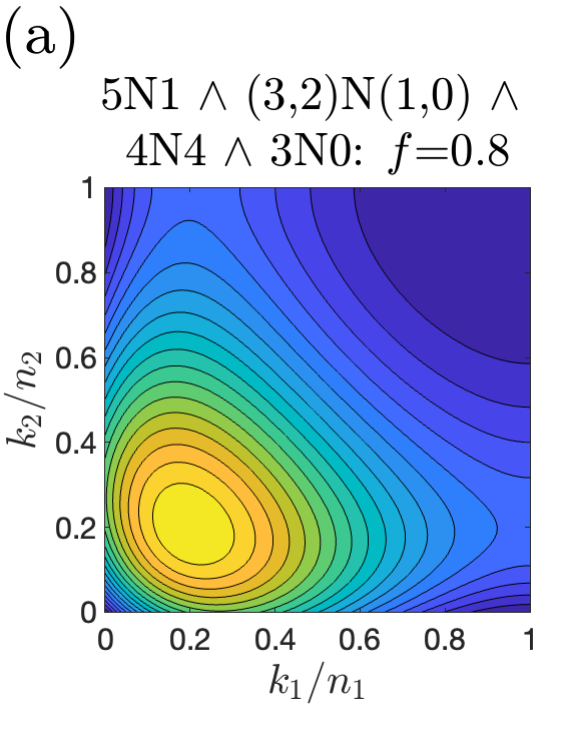}}
\includegraphics[width=0.397\linewidth,valign=t]{{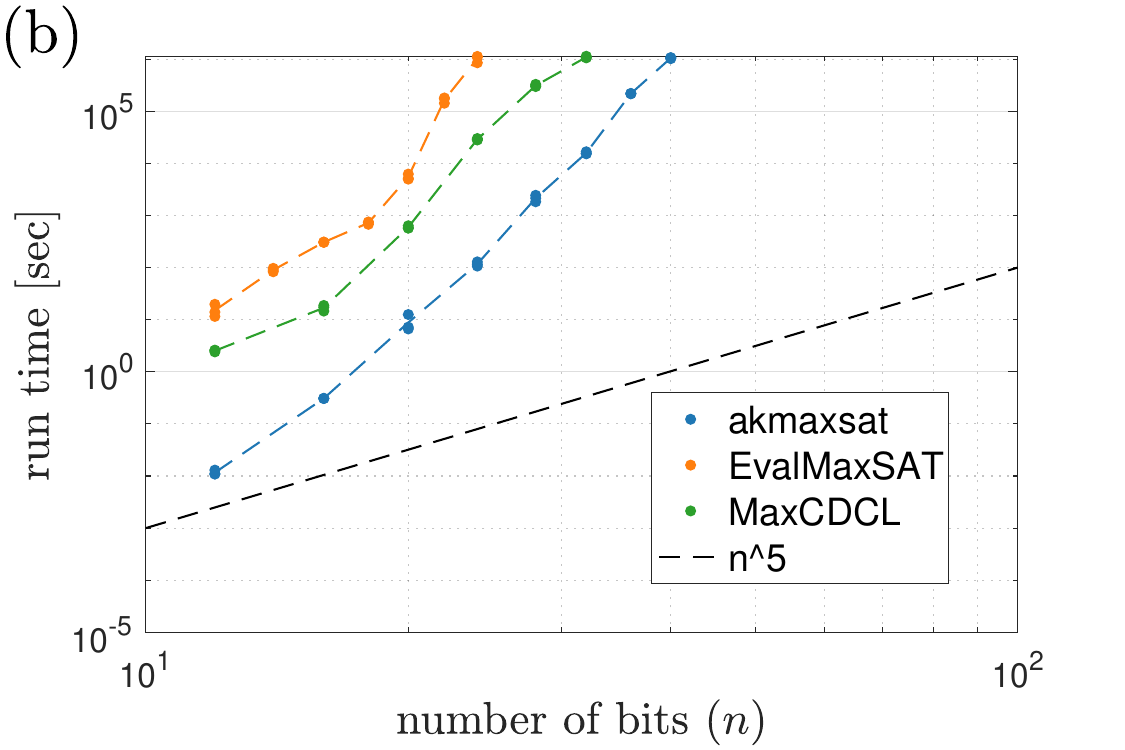}}
\includegraphics[width=0.397\linewidth,valign=t]{{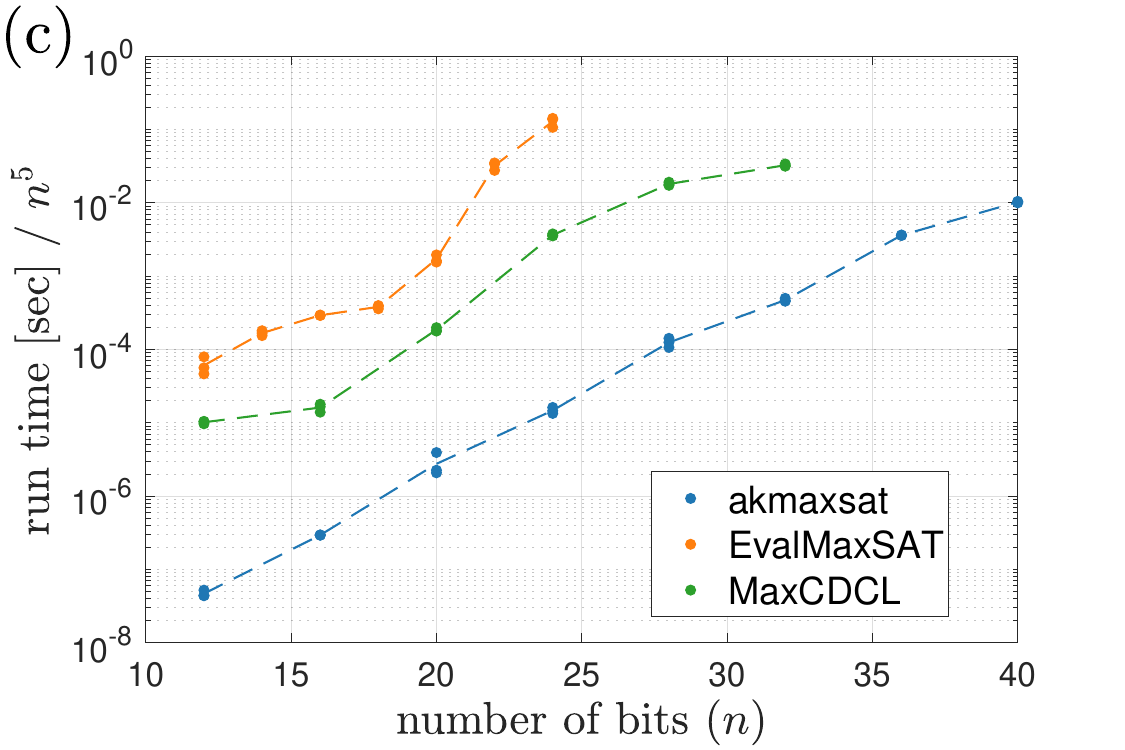}} \\
\includegraphics[width=0.206\linewidth,valign=t]{{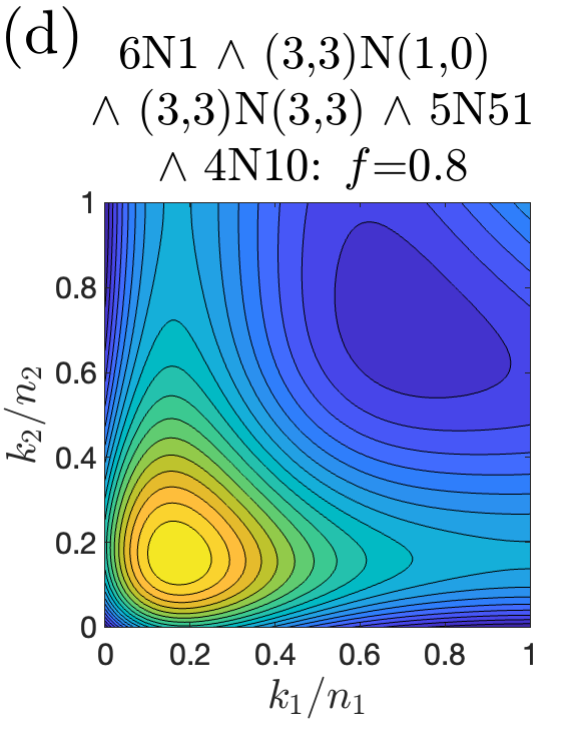}}
\includegraphics[width=0.397\linewidth,valign=t]{{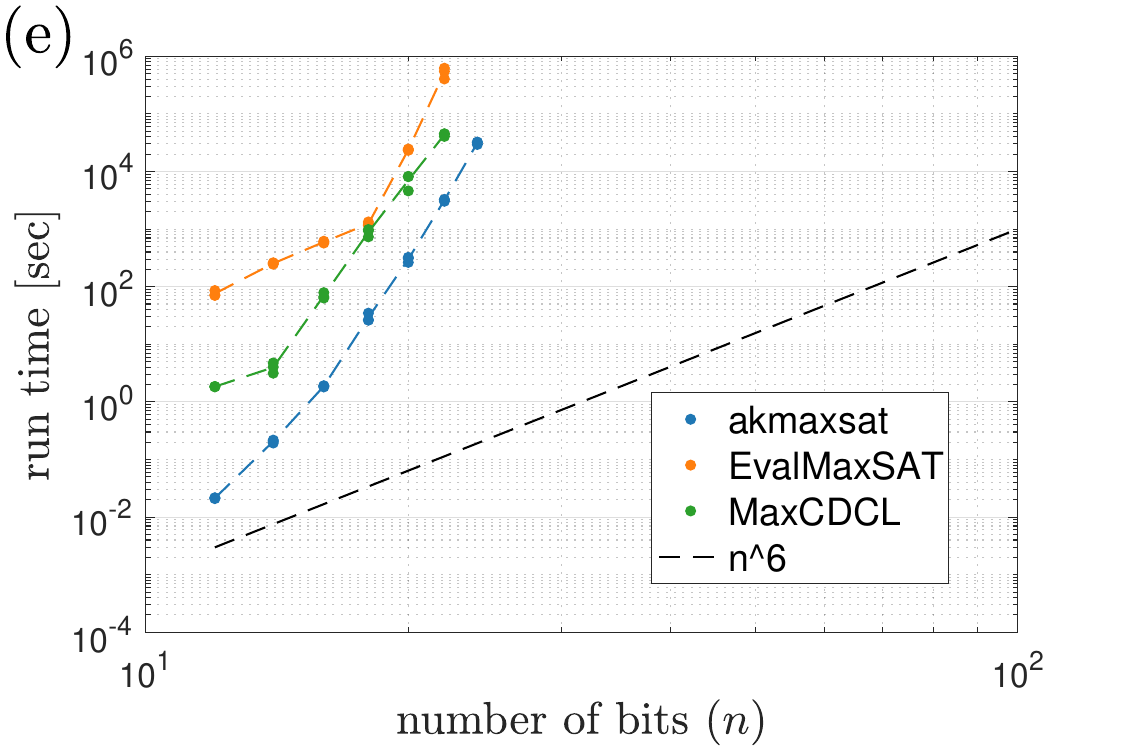}}
\includegraphics[width=0.397\linewidth,valign=t]{{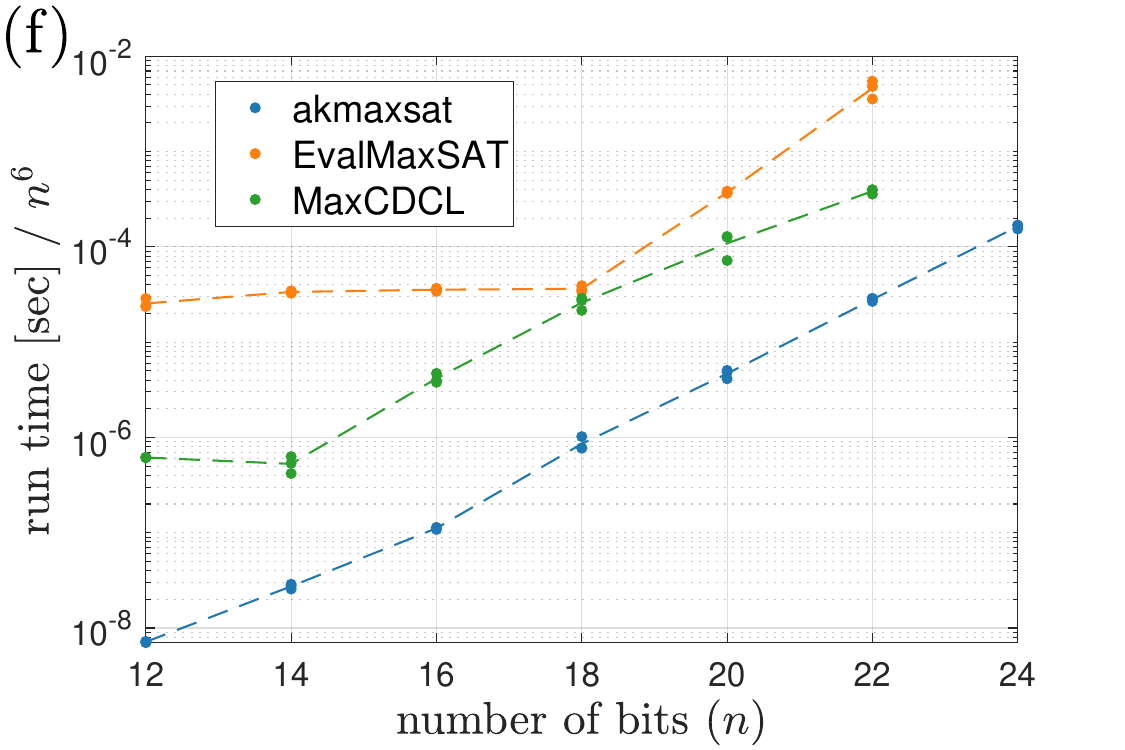}}
\caption{Performance of classical Max-SAT solvers on two families of near-$(S_{n/2})^2$-symmetric Max-$\ell$-SAT problems for $\ell = 5$ (top row) and $\ell=6$ (bottom row), constructed with clause sampling.
(a)(d) Contour plots of the energy landscapes of the two families of Max-$\ell$-SAT problems, as a function of the two subset Hamming distances to $\sv$.
(b)(e) Run times of various classical algorithms plotted as a function of number of bits $n$ on a log-log scale. (c)(f) Run times divided by $n^\ell$ plotted on a log-linear scale. Individual dots correspond to 3 problem instances at each $n$ with different clause sampling and choices of $\sv$, and dashed lines connect their averages.
For both problems, all three algorithms appear to require $\exp(n)$ run time in the large $n$ limit.
\label{fig:5SAT}
}
\end{figure}

We next consider near-symmetric Max-SAT problems. As an example, we consider $4\tN1 \wedge 3\tN 31 \wedge 2\tN0$ randomly sparsified by keeping each clause with probability $f$. The unsparsified $S_n$-symmetric cost function is visualized in Figure~\ref{fig:benchmark-SAT4}(a). Since these problems are not satisfiable by construction, SAT-only solvers such as \alg{kissat}, \alg{lstechMaple}, and \alg{WalkSATlm} are not applicable.
We benchmark the other three Max-SAT solvers on the near-$S_n$-symmetric $4\tN1 \wedge 3\tN 31 \wedge 2\tN0$ problem where $f=0.8$ fraction of the clauses are randomly included.
The run time results in Figure~\ref{fig:benchmark-SAT4}(b)(c) show that this problem is harder for \alg{MaxCDCL}, which seems to take exponential time. However, \alg{EvalMaxSAT} still appears to be able to solve these instances in $\Theta(n^4)$ which is linear in the problem size.
We also considered a family of near-$S_n$-symmetric Max-5-SAT problem in Figure~\ref{fig:benchmark-SAT4}(d)(e)(f). Nevertheless, the run time scaling of \alg{EvalMaxSAT} still appears approximately linear in the problem size $\Theta(n^5)$ despite the higher locality.

Finally, we consider Max-$\ell$-SAT problems that approximately exhibit $(S_{n/2})^2$ symmetry and hope that the additional local minima may pose further challenge to the classical algorithms.
The results for two families of such problems, $\ell=5$ and 6 respectively, are plotted in Figure~\ref{fig:5SAT}.
Both families of problems appears to stump all the Max-SAT solvers with exponential run time. Since our Theorems~\ref{thm:prod-sym} and \ref{thm:sparsified} imply that 1-step QAOA can solve them with $\Omega(1)$ probability using $O(n^\ell)$ quantum gates, these two families of near-$(S_{n/2})^2$-symmetric Max-$\ell$-SAT provides an example optimization problem that gives the QAOA an exponential speedup over the state-of-the-art classical Max-SAT solvers.

\section{Discussion}

We have studied the performance of the QAOA on symmetric and near-symmetric combinatorial optimization problems, and found evidence that even 1-step QAOA can outperform state-of-the-art classical algorithms and offer exponential speedups in some cases.
The symmetric optimization problems we consider are maximum constraint satisfaction problems (Max-CSPs) where the cost function is invariant under a permutation subgroup acting on the bits, relative to a solution bit string. These symmetries enable rigorous analysis of the QAOA, and we derive the closed form expression for its success probability, uncovering intriguing quantum mechanisms that allow the QAOA to quickly find the solution in many cases.

For example, when the cost function is purely a function of the Hamming distance to the solution and takes on mostly distinct values, we show that the 1-step QAOA has $\Omega(1/\sqrt{n})$ success probability in finding the solution. For cost functions that take on specific forms with multiple local minima that confuse classical algorithms, we show that 1-step QAOA can find the solution with $\Omega(1)$ probability.
This means with repeated measurements, the QAOA can find the solution with $O(\sqrt{n})$ or $O(1)$ queries to the cost function, using $\poly(n)$ time whenever the cost function has a polynomial-sized instantiation.

These results allow us to demonstrate significant quantum speedups by the QAOA over state-of-the-art classical algorithms on these symmetric and near-symmetric problems.
In the query complexity setting, we show that for $S_n$-symmetric problems, any classical algorithm provably requires at least $\Omega(n/\log n)$ queries to the cost function, even knowing in advance the symmetry and shape of the cost function. 
Depending on the problem, this is quadratically or superpolynomially worse compared to the quantum query complexity using 1-step QAOA.
In the more general setting,  we explicitly design families of symmetric and near-symmetric Max-SAT problems and study the performance of leading classical optimization algorithms.
These algorithms include simulated annealing and several state-of-the-art classical SAT and Max-SAT solvers that are recent winners of SAT and Max-SAT competitions, whose run time scaling we study through numerical benchmark on problems involving up to hundreds of bits.
Our results show that  for some of these problems, all of these classical algorithms take exponential time to find the solution. This implies an exponential separation between 1-step QAOA and general-purpose classical algorithms.

Our work opens up several intriguing future research directions.
Although symmetry might appear important for the QAOA to succeed, we have shown that our result is robust even when the symmetry is broken by randomly sampling the clauses. In particular, we show that the QAOA can succeed with $\Omega(1)$ probability even when taking a symmetric Max-$\ell$-SAT problem with $\Theta(n^\ell)$ clauses and sampling only a subset of $\Theta(n^2)$ clauses, which may also enable a practical implementation on near-term quantum computers.
Furthermore, our random sparsification is just one possible method of breaking the symmetry while maintaining high success probability for the QAOA. It would be interesting to explore other ways that an optimization problem can be approximately symmetric and still yield a quantum speedup.

While we have mostly focused on comparisons to state-of-the-art general-purpose classical algorithms for the sparsified near-symmetric SAT problems, it is likely that a tailored classical algorithm that knows the potential presence of a symmetry can solve these problems faster in some cases.
For example, a simple classical attack knowing how we constructed our instances in Section~\ref{sec:benchmark} is to flip all the bits after reaching a local minima in our near-$S_n$-symmetric problems, or to use community detection to identify the two subsets in our near-$(S_{n/2})^2$-symmetric problems and flip one or both subsets of bits at a local minimum; however, this line of attack may be thwarted by random constructions drawing from a large set of more complicated symmetries and energy landscapes, perhaps with help from cryptographic methods. 
There are also classical attacks that attempt to learn the hidden string from observing the sampled clauses, which would work on some of our constructed problems where clauses of a known locality encode the secret string as a unique satisfying solution; however, we believe this approach should not work for general, frustrated near-symmetric problems, such as the sparse LPN problem discussed in Remark~\ref{rem:LPN}.
We leave the construction of near-symmetric problems that are easy for QAOA while resistant to tailored classical attacks as future work.

We believe that studying optimization problems with symmetries enables better understanding of the unique behaviors of quantum algorithms and opens up new possibilities of superpolynomial quantum speedups. Our result has focused on two example symmetries, namely $S_n$ and $S_{n_1}\times S_{n_2}$ that permute the order of bit variables. As seen in our results, going from $S_n$ to $(S_{n/2})^2$ produces more complex problems with possibility of additional local minima that confuse classical algorithms and lead to an exponential quantum speedup over the best general-purpose classical solvers.
Therefore, it would be interesting to consider extensions of our work to other symmetry groups, including not only those that permute the bits but also those that act on the bit strings directly. We remark that the analysis of quantum algorithms like the QAOA is tractable as long as the number of orbits from the group action is polynomially bounded. 
There is a rich landscape of symmetric optimization problems awaiting exploration for large quantum speedups.

\subsection*{Acknowledgments}
We thank Edward Farhi and Sam Gutmann for providing helpful comments and discussions of classical attacks.
LZ thanks Song Mei for suggesting a more rigorous approach to approximate the binomial sums.
AM and LZ acknowledge funding from the European Research Council (ERC) under the European Union's Horizon 2020 research and innovation programme (grant agreement No. 817581).

\bibliographystyle{hunsrtnat}
\bibliography{refs}

\begin{thebibliography}{34}
\expandafter\ifx\csname natexlab\endcsname\relax\def\natexlab#1{#1}\fi
\expandafter\ifx\csname url\endcsname\relax
  \def\url#1{{\tt #1}}\fi

\bibitem[Farhi et~al.(2014)Farhi, Goldstone, and Gutmann]{farhi2014quantum}
Edward Farhi, Jeffrey Goldstone, and Sam Gutmann.
\newblock {A Quantum Approximate Optimization Algorithm}.
\newblock {\em arXiv preprint arXiv:1411.4028}, 2014.
\newblock URL \url{https://arxiv.org/abs/1411.4028}.

\bibitem[Pagano et~al.(2020)Pagano, Bapat, Becker, Collins, De, Hess, Kaplan,
  Kyprianidis, Tan, Baldwin, Brady, Deshpande, Liu, Jordan, Gorshkov, and
  Monroe]{Pagano2020QAOA}
Guido Pagano, Aniruddha Bapat, Patrick Becker, Katherine~S. Collins, Arinjoy
  De, Paul~W. Hess, Harvey~B. Kaplan, Antonis Kyprianidis, Wen~Lin Tan,
  Christopher Baldwin, Lucas~T. Brady, Abhinav Deshpande, Fangli Liu, Stephen
  Jordan, Alexey~V. Gorshkov, and Christopher Monroe.
\newblock Quantum approximate optimization of the long-range ising model with a
  trapped-ion quantum simulator.
\newblock {\em Proceedings of the National Academy of Sciences}, 117\penalty0
  (41):\penalty0 25396--25401, 2020.
\newblock URL \url{https://www.pnas.org/doi/abs/10.1073/pnas.2006373117}.

\bibitem[Harrigan et~al.(2021)Harrigan, Sung, Neeley, Satzinger, Arute,
  et~al.]{harrigan2021quantum}
Matthew~P Harrigan, Kevin~J Sung, Matthew Neeley, Kevin~J Satzinger, Frank
  Arute, et~al.
\newblock {Quantum approximate optimization of non-planar graph problems on a
  planar superconducting processor}.
\newblock {\em Nature Physics}, 17\penalty0 (3):\penalty0 332--336, 2021.

\bibitem[Ebadi et~al.(2022)Ebadi, Keesling, Cain, Wang, Levine, Bluvstein,
  Semeghini, Omran, Liu, Samajdar, Luo, Nash, Gao, Barak, Farhi, Sachdev,
  Gemelke, Zhou, Choi, Pichler, Wang, Greiner, Vuletic, and
  Lukin]{Ebadi2022quantum}
Sepehr Ebadi, Alexander Keesling, Madelyn Cain, Tout~T. Wang, Harry Levine,
  Dolev Bluvstein, Giulia Semeghini, Ahmed Omran, Jinguo Liu, Rhine Samajdar,
  Xiu-Zhe Luo, Beatrice Nash, Xun Gao, Boaz Barak, Edward Farhi, Subir Sachdev,
  Nathan Gemelke, Leo Zhou, Soonwon Choi, Hannes Pichler, Shengtao Wang, Markus
  Greiner, Vladan Vuletic, and Mikhail~D. Lukin.
\newblock {Quantum Optimization of Maximum Independent Set using Rydberg Atom
  Arrays}.
\newblock {\em Science}, 376\penalty0 (6598):\penalty0 1209--1215, 2022.

\bibitem[Shaydulin et~al.(2024)Shaydulin, Li, Chakrabarti, DeCross, Herman,
  Kumar, Larson, Lykov, Minssen, Sun, Alexeev, Dreiling, Gaebler, Gatterman,
  Gerber, Gilmore, Gresh, Hewitt, Horst, Hu, Johansen, Matheny, Mengle, Mills,
  Moses, Neyenhuis, Siegfried, Yalovetzky, and Pistoia]{Shaydulin2024QAOA}
Ruslan Shaydulin, Changhao Li, Shouvanik Chakrabarti, Matthew DeCross, Dylan
  Herman, Niraj Kumar, Jeffrey Larson, Danylo Lykov, Pierre Minssen, Yue Sun,
  Yuri Alexeev, Joan~M. Dreiling, John~P. Gaebler, Thomas~M. Gatterman,
  Justin~A. Gerber, Kevin Gilmore, Dan Gresh, Nathan Hewitt, Chandler~V. Horst,
  Shaohan Hu, Jacob Johansen, Mitchell Matheny, Tanner Mengle, Michael Mills,
  Steven~A. Moses, Brian Neyenhuis, Peter Siegfried, Romina Yalovetzky, and
  Marco Pistoia.
\newblock Evidence of scaling advantage for the quantum approximate
  optimization algorithm on a classically intractable problem.
\newblock {\em Science Advances}, 10\penalty0 (22):\penalty0 eadm6761, 2024.
\newblock URL \url{https://www.science.org/doi/abs/10.1126/sciadv.adm6761}.

\bibitem[Bravyi et~al.(2020)Bravyi, Kliesch, Koenig, and
  Tang]{bravyi2020obstacles}
Sergey Bravyi, Alexander Kliesch, Robert Koenig, and Eugene Tang.
\newblock Obstacles to variational quantum optimization from symmetry
  protection.
\newblock {\em Physical Review Letters}, 125\penalty0 (26):\penalty0 260505,
  2020.

\bibitem[Farhi et~al.(2020{\natexlab{a}})Farhi, Gamarnik, and
  Gutmann]{farhi2020quantumWhole1}
Edward Farhi, David Gamarnik, and Sam Gutmann.
\newblock {The quantum approximate optimization algorithm needs to see the
  whole graph: A typical case}.
\newblock {\em arXiv preprint arXiv:2004.09002}, 2020{\natexlab{a}}.

\bibitem[Farhi et~al.(2020{\natexlab{b}})Farhi, Gamarnik, and
  Gutmann]{farhi2020quantumWhole2}
Edward Farhi, David Gamarnik, and Sam Gutmann.
\newblock {The quantum approximate optimization algorithm needs to see the
  whole graph: Worst case examples}.
\newblock {\em arXiv preprint arXiv:2005.08747}, 2020{\natexlab{b}}.

\bibitem[Chou et~al.(2022)Chou, Love, Sandhu, and Shi]{chou2021limitations}
Chi-Ning Chou, Peter~J. Love, Juspreet~Singh Sandhu, and Jonathan Shi.
\newblock {Limitations of Local Quantum Algorithms on Random MAX-k-XOR and
  Beyond}.
\newblock In {\em 49th International Colloquium on Automata, Languages, and
  Programming (ICALP 2022)}, volume 229, pages 41:1--41:20, Dagstuhl, Germany,
  2022.
\newblock ISBN 978-3-95977-235-8.

\bibitem[Basso et~al.(2022{\natexlab{a}})Basso, Gamarnik, Mei, and
  Zhou]{basso2022performance}
Joao Basso, David Gamarnik, Song Mei, and Leo Zhou.
\newblock Performance and limitations of the {QAOA} at constant levels on large
  sparse hypergraphs and spin glass models.
\newblock In {\em 63rd Annual Symposium on Foundations of Computer Science
  (FOCS 2022)}, pages 335--343. IEEE, 2022{\natexlab{a}}.

\bibitem[Anshu and Metger(2023)]{anshu2023concentration}
Anurag Anshu and Tony Metger.
\newblock {Concentration Bounds for Quantum States and Limitations on the QAOA
  from Polynomial Approximations}.
\newblock In {\em 14th Innovations in Theoretical Computer Science Conference
  (ITCS 2023)}, volume 251, pages 5:1--5:8, 2023.
\newblock ISBN 978-3-95977-263-1.

\bibitem[Chen et~al.(2023)Chen, Huang, and Marwaha]{chen2023local}
Antares Chen, Neng Huang, and Kunal Marwaha.
\newblock Local algorithms and the failure of log-depth quantum advantage on
  sparse random {CSPs}.
\newblock {\em arXiv preprint arXiv:2310.01563}, 2023.

\bibitem[Basso et~al.(2022{\natexlab{b}})Basso, Farhi, Marwaha, Villalonga, and
  Zhou]{basso2021quantum}
Joao Basso, Edward Farhi, Kunal Marwaha, Benjamin Villalonga, and Leo Zhou.
\newblock {The Quantum Approximate Optimization Algorithm at High Depth for
  MaxCut on Large-Girth Regular Graphs and the Sherrington-Kirkpatrick Model}.
\newblock In {\em 17th Conference on the Theory of Quantum Computation,
  Communication and Cryptography (TQC 2022)}, volume 232, pages 7:1--7:21,
  2022{\natexlab{b}}.
\newblock ISBN 978-3-95977-237-2.

\bibitem[Boulebnane and Montanaro(2024)]{boulebnane2024SAT}
Sami Boulebnane and Ashley Montanaro.
\newblock {Solving Boolean Satisfiability Problems With The Quantum Approximate
  Optimization Algorithm}.
\newblock {\em PRX Quantum}, 5:\penalty0 030348, 2024.
\newblock URL \url{https://doi.org/10.1103/PRXQuantum.5.030348}.

\bibitem[van Dam et~al.(2001)van Dam, Mosca, and Vazirani]{vandam2001FOCS}
Wim van Dam, Michele Mosca, and Umesh Vazirani.
\newblock How powerful is adiabatic quantum computation?
\newblock In {\em 42nd Annual Symposium on Foundations of Computer Science
  (FOCS 2001)}, pages 279--287. IEEE, 2001.
\newblock URL \url{https://doi.org/10.1109/SFCS.2001.959902}.

\bibitem[Farhi et~al.(2002)Farhi, Goldstone, and Gutmann]{farhi2002QAAvsSA}
Edward Farhi, Jeffrey Goldstone, and Sam Gutmann.
\newblock Quantum adiabatic evolution algorithms versus simulated annealing.
\newblock {\em arXiv preprint arXiv:quant-ph/0201031}, 2002.

\bibitem[Reichardt(2004)]{reichardt2004adiabatic}
Ben~W. Reichardt.
\newblock The quantum adiabatic optimization algorithm and local minima.
\newblock In {\em Proceedings of the Thirty-Sixth Annual ACM Symposium on
  Theory of Computing (STOC 2004)}, page 502–510, New York, NY, USA, 2004.
  Association for Computing Machinery.
\newblock ISBN 1581138520.
\newblock URL \url{https://doi.org/10.1145/1007352.1007428}.

\bibitem[Muthukrishnan et~al.(2016)Muthukrishnan, Albash, and
  Lidar]{muthukrishnan2016}
Siddharth Muthukrishnan, Tameem Albash, and Daniel~A. Lidar.
\newblock Tunneling and speedup in quantum optimization for
  permutation-symmetric problems.
\newblock {\em Phys. Rev. X}, 6:\penalty0 031010, 2016.
\newblock URL \url{https://doi.org/10.1103/PhysRevX.6.031010}.

\bibitem[Zhou et~al.(2020)Zhou, Wang, Choi, Pichler, and Lukin]{ZhouQAOA}
Leo Zhou, Sheng-Tao Wang, Soonwon Choi, Hannes Pichler, and Mikhail~D. Lukin.
\newblock {Quantum Approximate Optimization Algorithm: Performance, Mechanism,
  and Implementation on Near-Term Devices}.
\newblock {\em Phys. Rev. X}, 10:\penalty0 021067, 2020.
\newblock URL \url{https://doi.org/10.1103/PhysRevX.10.021067}.

\bibitem[Farhi et~al.(2022)Farhi, Goldstone, Gutmann, and
  Zhou]{farhi2022quantum}
Edward Farhi, Jeffrey Goldstone, Sam Gutmann, and Leo Zhou.
\newblock The {Q}uantum {A}pproximate {O}ptimization {A}lgorithm and the
  {S}herrington-{K}irkpatrick {M}odel at {I}nfinite {S}ize.
\newblock {\em {Quantum}}, 6:\penalty0 759, 2022.

\bibitem[Boulebnane and Montanaro(2021)]{boulebnane2021predicting}
Sami Boulebnane and Ashley Montanaro.
\newblock {Predicting parameters for the Quantum Approximate Optimization
  Algorithm for MAX-CUT from the infinite-size limit}.
\newblock {\em arXiv preprint arXiv:2110.10685}, 2021.

\bibitem[Claes and van Dam(2021)]{claes2021instance}
Jahan Claes and Wim van Dam.
\newblock Instance independence of single layer quantum approximate
  optimization algorithm on mixed-spin models at infinite size.
\newblock {\em Quantum}, 5:\penalty0 542, 2021.

\bibitem[El~Alaoui et~al.(2023)El~Alaoui, Montanari, and Sellke]{AMS2021MaxCut}
Ahmed El~Alaoui, Andrea Montanari, and Mark Sellke.
\newblock Local algorithms for maximum cut and minimum bisection on locally
  treelike regular graphs of large degree.
\newblock {\em Random Structures \& Algorithms}, 63\penalty0 (3):\penalty0
  689–715, May 2023.
\newblock ISSN 1098-2418.
\newblock URL \url{http://dx.doi.org/10.1002/rsa.21149}.

\bibitem[TQC(2022)]{TQC22talk}
{The QAOA at High Depth for MaxCut on Large-Girth Regular Graphs and the
  Sherrington-Kirkpatrick Model}.
\newblock Talk presented at 17th Conference on the Theory of Quantum
  Computation, Communication and Cryptography (TQC'22), 2022.
\newblock URL \url{https://www.youtube.com/watch?v=yYmwfEYtKO4&t=5267s}.

\bibitem[Bennett et~al.(1997)Bennett, Bernstein, Brassard, and Vazirani]{BBBV}
Charles~H. Bennett, Ethan Bernstein, Gilles Brassard, and Umesh Vazirani.
\newblock Strengths and weaknesses of quantum computing.
\newblock {\em SIAM Journal on Computing}, 26\penalty0 (5):\penalty0
  1510--1523, 1997.
\newblock URL \url{https://doi.org/10.1137/S0097539796300933}.

\bibitem[Pinna and Viola(2019)]{PinnaViola2019}
Francesco Pinna and Carlo Viola.
\newblock {The saddle-point method in $\mathbb{C}^N$ and the generalized Airy
  functions}.
\newblock {\em Bull. Soc. Math. France}, 147:\penalty0 221--257, 2019.

\bibitem[Blum et~al.(2003)Blum, Kalai, and Wasserman]{Blum2003LPN}
Avrim Blum, Adam Kalai, and Hal Wasserman.
\newblock Noise-tolerant learning, the parity problem, and the statistical
  query model.
\newblock {\em J. ACM}, 50\penalty0 (4):\penalty0 506–519, July 2003.
\newblock ISSN 0004-5411.
\newblock URL \url{https://doi.org/10.1145/792538.792543}.

\bibitem[Pietrzak(2012)]{Pietrzak2012LPNcrypto}
Krzysztof Pietrzak.
\newblock Cryptography from learning parity with noise.
\newblock In {\em SOFSEM 2012: Theory and Practice of Computer Science}, pages
  99--114, Berlin, Heidelberg, 2012. Springer Berlin Heidelberg.

\bibitem[Applebaum et~al.(2010)Applebaum, Barak, and
  Wigderson]{Applebaum2010public}
Benny Applebaum, Boaz Barak, and Avi Wigderson.
\newblock Public-key cryptography from different assumptions.
\newblock In {\em Proceedings of the Forty-Second ACM Symposium on Theory of
  Computing (STOC 2010)}, page 171–180, New York, NY, USA, 2010. Association
  for Computing Machinery.
\newblock ISBN 9781450300506.
\newblock URL \url{https://doi.org/10.1145/1806689.1806715}.

\bibitem[Chen et~al.(2024)Chen, Shu, and Zhou]{chen2024sparselpn}
Xue Chen, Wenxuan Shu, and Zhaienhe Zhou.
\newblock {Algorithms for Sparse LPN and LSPN Against Low-noise}.
\newblock {\em arXiv preprint arXiv:2407.19215}, 2024.
\newblock URL \url{https://arxiv.org/abs/2407.19215}.

\bibitem[Wang et~al.(2016)Wang, Huang, Lee, and Chen]{wang16}
I-Hsiang Wang, Shao-Lun Huang, Kuan-Yun Lee, and Kwang-Cheng Chen.
\newblock Data extraction via histogram and arithmetic mean queries:
  Fundamental limits and algorithms.
\newblock In {\em 2016 IEEE International Symposium on Information Theory
  (ISIT)}, pages 1386--1390, 2016.

\bibitem[Gebhard et~al.(2022)Gebhard, Hahn-Klimroth, Kaaser, and
  Loick]{gebhard22}
Oliver Gebhard, Max Hahn-Klimroth, Dominik Kaaser, and Philipp Loick.
\newblock On the parallel reconstruction from pooled data.
\newblock In {\em 2022 IEEE International Parallel and Distributed Processing
  Symposium (IPDPS)}, pages 425--435, 2022.

\bibitem[Soleymani and Javidi(2024)]{shipra24}
Mahdi Soleymani and Tara Javidi.
\newblock {A Non-Adaptive Algorithm for the Quantitative Group Testing
  Problem}.
\newblock In {\em Proceedings of Thirty Seventh Conference on Learning Theory},
  volume 247 of {\em Proceedings of Machine Learning Research}, pages
  4574--4592. PMLR, 2024.
\newblock URL \url{https://proceedings.mlr.press/v247/soleymani24a.html}.

\bibitem[Chen and Wang(2017)]{Chen2017QGTnoise}
Wei-Ning Chen and I-Hsiang Wang.
\newblock Partial data extraction via noisy histogram queries: Information
  theoretic bounds.
\newblock In {\em 2017 IEEE International Symposium on Information Theory
  (ISIT)}, pages 2488--2492, 2017.

\end{thebibliography}

\end{document}